\documentclass[preprint,12pt]{elsarticle}

\usepackage{graphicx}
\usepackage{natbib}

\usepackage{algorithm}
\usepackage{algorithmic}
\usepackage{comment}

\usepackage{amsthm}
\usepackage{amssymb}
\usepackage{bbm}
\usepackage{amsmath}
\usepackage{amsfonts}
\usepackage{mathtools}
\usepackage{framed}
\usepackage{pifont}
\usepackage{xspace}

\usepackage{hyperref}

\usepackage{float}
\usepackage{tikz}
\usetikzlibrary{decorations.pathreplacing,patterns,positioning}

\usepackage[dvipsnames]{xcolor}

\usepackage{multirow}

\newcommand*{\task}[4]{
	\draw[rounded corners] (#3-#2,#4+0.1) rectangle (#3,#4+0.6);
	\draw (#3-#2/2,#4+0.35) node {#1};
}

\newcommand*{\taskDashed}[4]{
	\draw[rounded corners,dashed] (#3-#2,#4+0.1) rectangle (#3,#4+0.6);
	\draw (#3-#2/2,#4+0.35) node {#1};
}

{\itshape}{\rmfamily}

{\itshape}{\rmfamily}

\newcommand{\sigmaD}{$\Sigma $D\xspace}

\newtheorem{innercustomgeneric}{\customgenericname}
\providecommand{\customgenericname}{}

\newtheorem{theorem}{Theorem}
\newtheorem{lemma}{Lemma}

\newtheorem{obs}{Observation}

\newcommand{\setTask}{\ensuremath{\mathcal{J}}\xspace}
\newcommand{\nbTask}{\ensuremath{n}\xspace}
\newcommand{\setVoter}{\ensuremath{\mathcal{V}}\xspace}
\newcommand{\nbVoter}{\ensuremath{v}\xspace}
\newcommand{\pref}{\ensuremath{V}\xspace}
\newcommand{\prefVoter}[1]{\ensuremath{\pref_{#1}}\xspace}
\newcommand{\prefProfile}{\ensuremath{\mathcal{P}}\xspace}

\newcommand{\sched}{\ensuremath{S}\xspace}
\newcommand{\setPerm}{\ensuremath{\mathcal{X}^{\setTask}}\xspace}

\newcommand{\completionTime}{\ensuremath{C}\xspace}
\newcommand{\completionTimeTask}[1]{\ensuremath{\completionTime}_{#1}\xspace}
\newcommand{\completionTimeTaskSchedule}[2]{\ensuremath{\completionTimeTask{#1}(#2)}\xspace}
\newcommand{\dueDate}{\ensuremath{d}\xspace}
\newcommand{\dueDateTaskVoter}[2]{\ensuremath{\dueDate_{#1,#2}}\xspace}

\newcommand{\procTime}{\ensuremath{p}\xspace}
\newcommand{\procTimeTask}[1]{\ensuremath{\procTime_{#1}}\xspace}

\newcommand{\decprob}[3]{\textsc{#1}:\\\textbf{Input:} #2\\\textbf{Question:} #3}

\newcommand{\setInt}{\ensuremath{\mathcal{X}}\xspace}
\newcommand{\nameInt}{\ensuremath{x}\xspace}
\newcommand{\nbTriplet}{\ensuremath{q}\xspace}
\newcommand{\nameTriplet}{\ensuremath{T}\xspace}
\newcommand{\sumTriplet}{\ensuremath{B}\xspace}

\newcommand{\setIntTask}{\ensuremath{\mathcal{T}}\xspace}
\newcommand{\nameTask}{\ensuremath{t}\xspace}

\newcommand{\setSeparator}[1]{\ensuremath{A_{#1}}\xspace}
\newcommand{\nameSeparator}{\ensuremath{a}\xspace}
\newcommand{\setLeft}{\ensuremath{\mathcal{L}}\xspace}
\newcommand{\nameLeft}{\ensuremath{l}\xspace}
\newcommand{\setRight}{\ensuremath{\mathcal{R}}\xspace}
\newcommand{\nameRight}{\ensuremath{r}\xspace}
\newcommand{\setCentral}{\ensuremath{\mathcal{M}}\xspace}
\newcommand{\nameCentral}{\ensuremath{m}\xspace}

\newcommand{\factorK}{\ensuremath{K}\xspace}
\newcommand{\sizeSeparator}{\ensuremath{O}\xspace}
\newcommand{\sizeLargeSeparator}{\ensuremath{O}'\xspace}
\newcommand{\sumTripletNew}{\ensuremath{B'}\xspace}

\newcommand{\targetSum}{\ensuremath{Z}\xspace}

\newcommand{\sigmaDdecprob}{\textsc{\sigmaD-Dec}\xspace}

\newcommand{\intuition}[1]{\textit{#1}}

\begin{document}

\title{Two NP-hard Extensions of the Spearman Footrule even for a Small Constant Number of Voters}

\author{Martin Durand, martin.durand@lip6.fr \\
       Sorbonne Université, LIP6, CNRS,\\
       4 Place Jussieu, 75005, Paris, France}


\begin{frontmatter}







\begin{abstract}

The Spearman footrule is a voting rule that takes as input voter preferences expressed as rankings. It outputs a ranking that minimizes the sum of the absolute differences between the position of each candidate in the ranking and in the voters’ preferences. In this paper, we study the computational complexity of two extensions of the Spearman footrule when the number of voters is a small constant. The first extension, introduced by Pascual et al. (2018), arises from the collective scheduling problem and treats candidates, referred to as tasks in their model, as having associated lengths. The second extension, proposed by Kumar and Vassilvitskii (2010), assigns weights to candidates; these weights serve both as lengths, as in the collective scheduling model, and as coefficients in the objective function to be minimized. Although computing a ranking under the standard Spearman footrule is polynomial-time solvable, we demonstrate that the first extension is NP-hard with as few as 3 voters, and the second extension is NP-hard with as few as 4 voters. Both extensions are polynomial-time solvable for 2 voters.

\end{abstract}


\end{frontmatter}

\section{Introduction}
\label{sec:introduction}

In this paper, we analyze the computational complexity of a voting rule designed for the collective scheduling framework~\citep{pascual2018collective, durand2022collective}. Under this framework, a group of individuals, referred to as voters, express their preferences over the sequential execution of shared tasks on a single common machine. This setting models several real-world scenarios. For instance, consider a public infrastructure project, such as expanding a city’s subway system with multiple new lines or constructing new cycle paths. Given constraints on workforce, machinery, and annual budgets, these projects are typically decomposed into sequential phases of varying durations, e.g., longer cycle paths require longer construction phases. A possible way to reach a decision on the execution order of these phases is to solicit the preferred order of the concerned individuals, or their representatives. This raises the following question: Given the preferences of all individuals, how can we construct a collective schedule that maximizes overall satisfaction?

This problem generalizes the consensus ranking problem~\cite{brandt2016handbook}. Specifically, in the consensus ranking problem, a set of $\nbVoter$ voters each provide a ranking of $\nbTask$ candidates (or items). The objective is then to derive a consensus ranking from these preferences.
To illustrate the connection, suppose that all tasks in our problem have unit length. In this case, each task can be treated as a candidate, and a schedule can be interpreted as a ranking of these candidates. Consequently, computing a collective schedule reduces to computing a consensus ranking.
In this paper, we examine the computational complexity of the \sigmaD voting rule, as introduced by~\citeauthor{pascual2018collective}~\cite{pascual2018collective}. Formally, the \sigmaD rule takes as input the preferences of voters over the execution order of a set of shared tasks, expressed as permutations. It outputs a schedule $\sched$ on a single machine, that is, a permutation of the tasks, that minimizes the sum, over all tasks, of the absolute differences between the completion time of each task in $\sched$ and its completion time in the voters’ preferred schedules. The authors observe that the rule \sigmaD generalizes both the Spearman footrule~\citep{diaconis1977spearman}. Indeed, the Spearman footrule returns a permutation that minimizes the absolute difference between the position of each candidate in the returned permutation and in the preference of each voter, also called \emph{Spearman distance}. If all task are of length 1, the position equals the completion time.

\medskip

\noindent{\bf Related work.}
This work examines the computational complexity of a consensus ranking problem, situating it within the field of computational social choice~\citep{brandt2016handbook}. 
The classic Spearman footrule~\citep{diaconis1977spearman} can be seen as an assignment problem: There are $\nbTask$ candidates and $\nbTask$ slots, the goal is to assign each candidate to a slot. This problem is solvable in polynomial time using the Hungarian algorithm~\citep{Tomizawa71,EdmondsK72}. 
The Spearman footrule has been studied through its link with a well-known rule: The Kemeny rule~\cite{kemeny1959mathematics}. The Kemeny rule returns a permutation minimizing the number of pairwise disagreement with the preferences of the voters, also called \emph{Kendall-tau distance}. Computing a permutation minimizing the Kendall-tau distance with the preference profile is an NP-hard problem, even with as few as 4 voters~\cite{bartholdi1989voting}. However,~\citeauthor{diaconis1977spearman}~\cite{diaconis1977spearman} show that the Spearman distance of a permutation with a preference profile is at most twice its Kendall-tau distance with the same preference profile, showing that the Spearman rule returns a permutation which is 2-approximate for Kendall-tau distance minimization.

The collective scheduling problem has been introduced by~\citeauthor{pascual2018collective}~\cite{pascual2018collective}. They propose the \sigmaD rule as an extension of the absolute deviation criterion from scheduling theory~\cite{brucker1999scheduling}. Given a due date, the deviation of a task in a schedule is defined as the absolute difference between the completion time of the task in the schedule and the due date. In the non-collective scheduling problem, each task has a single due date and the goal is to compute a schedule minimizing the sum of the deviations of the tasks. This problem has been shown to be NP-hard~\cite{wan2013single}.

\citeauthor{pascual2018collective}~\cite{pascual2018collective} also introduce a weighted variant of the Condorcet principle, called the PTA Condorcet principle (where PTA stands for ``Processing Time Aware''). 
\citeauthor{durand2022collective}~\cite{durand2022collective} define an extension of the Kemeny rule based on PTA Condorcet consistency. They also show that the \sigmaD rule solves an NP-hard problem in the general case. Their reduction however relies on a large number of voters. Our contribution strengthens this result: We demonstrate NP-hardness for as few as 3 voters, ruling out parameterized approaches based solely on the number of voters.

\citeauthor{kumar2010generalized}~\cite{kumar2010generalized} study generalizations of the Spearman footrule and of the Kemeny rule~\citep{kemeny1959mathematics}, and show that these generalizations respect a weaker version of~\citeauthor{diaconis1977spearman}'s result. Notably, they introduce candidate weights, which define both the space a candidate occupies, i.e., a candidate of weight 3 requires 3 slots, and its importance in the objective function, i.e., the weight multiplies the candidate’s contribution to the total cost. In contrast, the collective scheduling setting assumes that all tasks have equal importance, with only their lengths varying. 
\citeauthor{kumar2010generalized}~\cite{kumar2010generalized} do not investigate the complexity of their extension of the  Spearman footrule, only noting that the assignment approach no longer applies. 
We provide complexity results for both generalizations of Spearman footrule: (1) the one in which candidates only have a size, i.e., in the collective schedule framework; and (2) the one in which the candidates have weights, encoding both their size and importance in the objective function. \\



\noindent{\bf Our contribution.}
We provide two polynomial time reductions using a small constant number of voters, namely 4 and 3, showing NP-hardness of two preference aggregation problems. The first problem consists in finding a collective schedule according to the \sigmaD rule. We show that this problem is NP-hard, even with as few as 3 voters. The second problem consists in returning a ranking minimizing the extended Spearman distance with candidate weights, as defined by~\citeauthor{kumar2010generalized}~\cite{kumar2010generalized}. For this second problem, only the reduction with 4 voters directly holds. 
We also show that these two problems can be solved in polynomial time when there are only 2 voters.

\medskip

\noindent{\bf Organization of the paper.}
Section~\ref{sec:preliminaries} introduces the collective scheduling framework as well as the \sigmaD voting rule. 
In Section~\ref{sec:complexity}, we show that the problem solved by \sigmaD is easy when there are 2 voters but NP-hard even when there are as few as 3 voters. We also provide another reduction with 4 voters, which holds in the case were tasks have weights.
We conclude the paper with a discussion in Section ~\ref{sec:conclu}. 

\section{Preliminaries}
\label{sec:preliminaries}

\subsection{Notations.}
Let $[\nbTask]=\{1,\dots,\nbTask\}$.
Let $\setTask\!=\![\nbTask]$ be the set of tasks. Each task $i\!\in\!\setTask$ has a length (or processing time) $\procTimeTask{i}$. 
We do not consider idle times between the tasks, and preemption is not allowed: A schedule of the tasks is thus a permutation of the tasks of $\setTask$. We denote by \setPerm the set of all possible schedules. 
We denote by $\setVoter\!=\![\nbVoter]$ the set of $\nbVoter$ voters. Each voter $k\!\in\! \setVoter$ expresses her favorite schedule $\prefVoter{k} \!\in\!\setPerm$ of the tasks in $\setTask$. The preference profile, \prefProfile, is the set of the favorite schedules of the $\nbVoter$ voters: $\prefProfile=\{\prefVoter{1}, \dots, \prefVoter{\nbVoter}\}$. 
Given a schedule \sched, we denote by \completionTimeTaskSchedule{a}{\sched} the completion time of task $a$ in \sched. We denote by \dueDateTaskVoter{a}{k} the completion time of task $a$ in the preferred schedule of voter $k$ (i.e., $\dueDateTaskVoter{a}{k}\!=\!\completionTimeTaskSchedule{a}{\prefVoter{k}}$) -- here $\dueDate$ stands for ``due date'' as this completion time can be seen as a due date of voter $a$ for task $k$. Given two tasks $a$ and $b$ and a schedule \sched, we denote by $a\succ_{\sched}b$ the fact that task $a$ is scheduled before task $b$ in schedule $\sched$.  We will also use this notation to describe a schedule. For example, the 3 tasks schedule $\sched$ in which task $a$ is scheduled in first position, followed by task $b$, and then task $c$ is described as $a\succ_{\sched}b \succ_{\sched}c$. 

An \emph{aggregation rule} is a mapping $r\!:\!(\setPerm)^v\!\rightarrow\!\setPerm$ that associates a schedule $\sched$ -- the consensus schedule -- to any preference profile $\prefProfile$. 

\medskip

The \sigmaD rule 
is an extension of the Absolute Deviation (D) scheduling metric \citep{brucker1999scheduling}. This metric measures the deviation between a schedule \sched and a set of given due dates for the tasks of the schedule. It sums, over all the tasks, the absolute value of the difference between the completion time of task $i$ in $\sched$ and its due date. 
By considering the completion time \dueDateTaskVoter{i}{k} of task $i$ in the preferred schedule $\prefVoter{k}$ as a due date given by voter $k$ for task $i$, we express the deviation $D(\sched, \prefVoter{k})$ between schedule \sched and schedule \prefVoter{k} as $D(\sched,\prefVoter{k})=\sum_{i \in \setTask} |\completionTimeTaskSchedule{i}{\sched}-\dueDateTaskVoter{i}{k}|$.
By summing over all the voters, we obtain a metric $D(\sched, \prefProfile)$ measuring the deviation between a schedule \sched and a preference profile \prefProfile:

\begin{equation}
    D(\sched,\prefProfile)=\sum\limits_{k \in \setVoter} \sum\limits_{i \in \setTask} |\completionTimeTaskSchedule{i}{\sched}-\dueDateTaskVoter{i}{k}|
\label{eq:formule_sigmaD}
\end{equation}

The \sigmaD rule returns a schedule $\sched^*$ minimizing the deviation with the preference profile $\prefProfile$: $D(\sched^*,\prefProfile)\!=\!\min_{\sched \in \setPerm}D(\sched, \prefProfile)$. 
We will call \emph{deviation of the task} $i$ in schedule \sched the value $\sum_{k \in \setVoter} |\completionTimeTaskSchedule{i}{\sched}-\dueDateTaskVoter{i}{k}|$. This corresponds to the cost of scheduling task $i$ at its place in schedule \sched.\\

This rule was introduced (but not studied) by Pascual et al. \citep{pascual2018collective}, where the authors observed that if tasks have unitary lengths this rule minimizes the Spearman distance. The Spearman distance is a well-known aggregation rule for the consensus ranking problem. Given a ranking $\sched$ and a preference profile $\prefProfile$, it is defined as $S(\sched,\prefProfile)\!=\!\sum_{k \in \setVoter} \sum_{i \in \setTask} |pos_i(\sched)-pos_i(\prefVoter{k})|$, where $pos_i(\sched)$ is the position of item $i$ in ranking $\sched$ (which means in our case the completion time of task $i$ in schedule \sched if items are unitary tasks). 

We also express the objective function of the Spearman footrule with weighted candidates \cite{kumar2010generalized} in the collective scheduling framework:

\begin{equation}
    D_w(\sched,\prefProfile)=\sum\limits_{k \in \setVoter} \sum\limits_{i \in \setTask} \procTimeTask{i}|\completionTimeTaskSchedule{i}{\sched}-\dueDateTaskVoter{i}{k}|
\label{eq:formule_sigmaD_weighted}
\end{equation}

\medskip

\section{Computational Complexity.}
\label{sec:complexity}

We start this section with a general remark: The deviation of a task is minimized when it is scheduled to complete at its \emph{median completion time} in the preference profile if the number of voters is odd, or within its \emph{median completion interval} in the preference profile, if the number of voters is even. Indeed, consider that a task completes at its median completion time, moving it later by one unit of time increases the deviation by one unit of time for at least $\nbVoter/2$ voters while it may decrease the deviation by one unit for at most $\nbVoter/2$. Similarly, if the number of voters is even, consider the interval between the $\nbVoter/2^{th}$ smallest completion time and the $\nbVoter/2+1^{th}$ smallest completion time in the preference profile, this is what we call the median completion interval. As long as the task completes within this interval, moving it by one unit of time increases the deviation for exactly $\nbVoter/2$ voters and decreases it for exactly $\nbVoter/2$ voters, however, if it goes out of the interval, then it increases the deviation of strictly more than $\nbVoter/2$ voters and decreases the deviation for at strictly less than $\nbVoter/2$ voters.\\

The decision version of our problem, that we will call {\sigmaDdecprob} is the following one:\\
\noindent\decprob{\sigmaDdecprob}{A set $\setTask=\{1,\dots,\nbTask\}$ of \nbTask tasks, such that $\forall i \in [\nbTask]$ task $\nameTask_i$ has processing time $\procTime_i$; a set $\prefProfile=\{\prefVoter{1},\dots,\prefVoter{\nbVoter}\}$ of $\nbVoter$ schedules of tasks \setTask; an integer \targetSum.}{Is there a schedule \sched\ of the tasks of \setTask\ on one machine such that $D(\sched,\prefProfile) \leq \targetSum$?}\\

When there are exactly two voters, the problem is easy to solve: We return one of the two schedules in the preference profile. Indeed, in such a schedule, each task is in its median completion interval, meaning that no schedule could have a total deviation strictly lower.\\ 
Note that this applies to the weighted case as well. Indeed the weights merely multiply the deviation of each task, if the deviation of each task is minimized, then so is the sum, weighted or not.\\
When there are strictly more than two voters, the \sigmaDdecprob problem is NP-hard, as shown later. 

We start with a reduction with 4 voters. 

\begin{theorem}
The \sigmaDdecprob problem is NP-hard, even when there are only 4 voters. 
\end{theorem}

\begin{proof}

For both the three voters case and the four voters case, we reduce the NP-complete \textsc{3-Partition} problem to our problem. 

\noindent\decprob{3-Partition}{A set $\setInt=\{\nameInt_1,\dots,\nameInt_{3\nbTriplet}\}$ of $3\nbTriplet$ integers, an integer \sumTriplet.}{Is there a partition of \setInt\ into \nbTriplet\ triplets $\nameTriplet_1$ to $\nameTriplet_\nbTriplet$ such that $\forall i \in [\nbTriplet]$, the sum of the integers in $\nameTriplet_i$ is exactly \sumTriplet?}\\

We use the fact that the deviation of a single task is minimized when it is completed at its median completion time. 
Therefore, if a solution in which each task completes at its median completion time exists, then it minimizes the total absolute deviation.
We create an instance for which such a solution exists if and only if the \textsc{3-Partition} instance is a yes-instance. 

Starting from the \textsc{3-Partition} problem, we create one task for each integer, with a processing time equal to the value of the integer, and we group these tasks into a set $\setTask$. We then create 4 sets $\mathcal{C}_1$ to $\mathcal{C}_4$, each containing as many unit tasks as the total sum of the integers (i.e., the total processing time of $\setTask$), so that each of these sets has the same total processing time. Finally, we create separator tasks of length 1, $\nameSeparator_1$ to $\nameSeparator_{\nbTriplet-1}$, where $\nbTriplet$ denotes the number of triplets to be formed in the \textsc{3-Partition} instance. The voters' preferences are illustrated in Figure~\ref{fig:pref_sketch}.

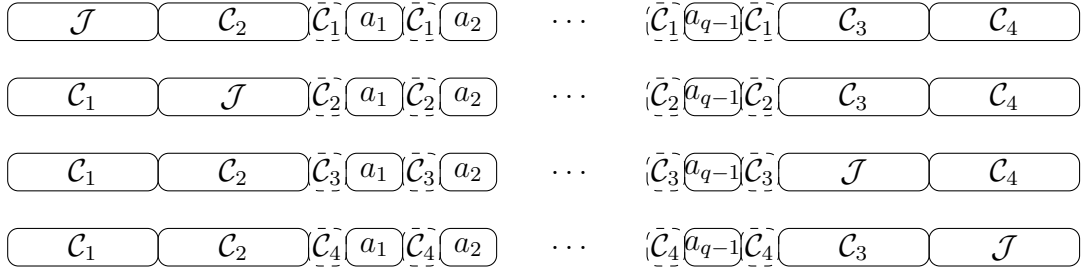
\begin{figure}[h!]
    \centering
    \begin{tikzpicture}

    \task{\setTask}{2}{2}{3}
    \task{$\mathcal{C}_2$}{2}{4}{3}
    \taskDashed{$\mathcal{C}_1$}{0.5}{4.5}{3}
    \task{$\nameSeparator_{1}$}{0.75}{5.25}{3}
    \taskDashed{$\mathcal{C}_1$}{0.5}{5.75}{3}
    \task{$\nameSeparator_{2}$}{0.75}{6.5}{3}

    \node at (7.5,3.35) {$\dots$};

    \taskDashed{$\mathcal{C}_1$}{0.5}{9}{3}
    \task{$\nameSeparator_{\nbTriplet-1}$}{0.75}{9.75}{3}
    \taskDashed{$\mathcal{C}_1$}{0.5}{10.25}{3}
    \task{$\mathcal{C}_3$}{2}{12.25}{3}
    \task{$\mathcal{C}_4$}{2}{14.25}{3}

    \task{$\mathcal{C}_1$}{2}{2}{2}
    \task{\setTask}{2}{4}{2}
    \taskDashed{$\mathcal{C}_2$}{0.5}{4.5}{2}
    \task{$\nameSeparator_{1}$}{0.75}{5.25}{2}
    \taskDashed{$\mathcal{C}_2$}{0.5}{5.75}{2}
    \task{$\nameSeparator_{2}$}{0.75}{6.5}{2}

    \node at (7.5,2.35) {$\dots$};

    \taskDashed{$\mathcal{C}_2$}{0.5}{9}{2}
    \task{$\nameSeparator_{\nbTriplet-1}$}{0.75}{9.75}{2}
    \taskDashed{$\mathcal{C}_2$}{0.5}{10.25}{2}
    \task{$\mathcal{C}_3$}{2}{12.25}{2}
    \task{$\mathcal{C}_4$}{2}{14.25}{2}

    \task{$\mathcal{C}_1$}{2}{2}{1}
    \task{$\mathcal{C}_2$}{2}{4}{1}
    \taskDashed{$\mathcal{C}_3$}{0.5}{4.5}{1}
    \task{$\nameSeparator_{1}$}{0.75}{5.25}{1}
    \taskDashed{$\mathcal{C}_3$}{0.5}{5.75}{1}
    \task{$\nameSeparator_{2}$}{0.75}{6.5}{1}

    \node at (7.5,1.35) {$\dots$};

    \taskDashed{$\mathcal{C}_3$}{0.5}{9}{1}
    \task{$\nameSeparator_{\nbTriplet-1}$}{0.75}{9.75}{1}
    \taskDashed{$\mathcal{C}_3$}{0.5}{10.25}{1}
    \task{$\setTask$}{2}{12.25}{1}
    \task{$\mathcal{C}_4$}{2}{14.25}{1}

    \task{$\mathcal{C}_1$}{2}{2}{0}
    \task{$\mathcal{C}_2$}{2}{4}{0}
    \taskDashed{$\mathcal{C}_4$}{0.5}{4.5}{0}
    \task{$\nameSeparator_{1}$}{0.75}{5.25}{0}
    \taskDashed{$\mathcal{C}_4$}{0.5}{5.75}{0}
    \task{$\nameSeparator_{2}$}{0.75}{6.5}{0}

    \node at (7.5,0.35) {$\dots$};

    \taskDashed{$\mathcal{C}_4$}{0.5}{9}{0}
    \task{$\nameSeparator_{\nbTriplet-1}$}{0.75}{9.75}{0}
    \taskDashed{$\mathcal{C}_4$}{0.5}{10.25}{0}
    \task{$\mathcal{C}_3$}{2}{12.25}{0}
    \task{$\setTask$}{2}{14.25}{0}
        
    \end{tikzpicture}
    \caption{Preferences of the voters in the four voters reduction. Dashed tasks means that there are as many unit tasks as the size of a target triplet in the 3-Partition instance. Tasks from all sets are always scheduled in the same order by the voters.}
    \label{fig:pref_sketch}
\end{figure}

We finally set \targetSum to be: 
\begin{equation*}
    \begin{array}{rcl}
        \targetSum & = & 2(\nbTriplet\sumTriplet)^2+\sum_{i=0}^{q-1} i\sumTriplet\\
         & + & (\nbTriplet\sumTriplet)^2 + \sum_{i=0}^{q-1} i\sumTriplet\\
         & + & (\nbTriplet\sumTriplet)^2 + \sum_{i=0}^{q-1} i\sumTriplet\\
         & + & 2(\nbTriplet\sumTriplet)^2+\sum_{i=0}^{\nbTriplet-1} i\sumTriplet\\
         & + & 3\nbTriplet(\nbTriplet\sumTriplet+\nbTriplet-1+2\nbTriplet\sumTriplet)
    \end{array}
\end{equation*}

Intuitively, each of these rows corresponds to the deviation obtained when one set of tasks is scheduled in their median completion interval. Consider the set $\mathcal{C}_1$, if a task from $\mathcal{C}_1$ is scheduled at its median completion time, it will be on time for the last 3 voters but early for the first voter. Specifically, the very first task of $\mathcal{C}_1$ will be early by 2 times $\nbTriplet\sumTriplet$, because the sets \setTask and $\mathcal{C}_2$ are scheduled before the tasks of $\mathcal{C}_1$ in the preference of the first voter. Now, note that the first $\sumTriplet$ tasks are not delayed more in the preference of the first voter. However, groups of $\sumTriplet$ tasks are delayed by the separator tasks, i.e., the second group by $1$ unit of time, the third by 2 units, and so on. This gives us $2(\nbTriplet\sumTriplet)^2$, which is $2(\nbTriplet\sumTriplet)$ for each of the $(\nbTriplet\sumTriplet)$ tasks plus $\sum_{i=0}^{q-1} i\sumTriplet$ as the first group is delayed by 0 separator tasks, the second by one and so on until the last group of \sumTriplet which is delayed by $\nbTriplet-1$ units of time. The same reasoning holds symmetrically for $\mathcal{C}_4$.

The same idea can be applied to $\mathcal{C}_2$ but each task is only early by $(\nbTriplet\sumTriplet)$ instead of twice that for $\mathcal{C}_1$. Once again, this holds symmetrically for $\mathcal{C}_3$.

The separator tasks are scheduled exactly at the same time in all preferences, therefore if they are also scheduled at this time in a solution, their deviation is exactly 0.

Finally, for the tasks of \setTask, we start by noting that, if a task is scheduled within the median completion time interval, then it is late for the first voter by exactly the same amount of time than for the second voter, plus $\nbTriplet\sumTriplet$. The same observation holds symmetrically for the third and fourth voter. Now consider the first task of the set, if it is scheduled at the very beginning of the interval, it is on time for the second, and early by $\nbTriplet\sumTriplet+\nbTriplet-1$ for the third voter, this gives us a total of $3\nbTriplet(\nbTriplet\sumTriplet+\nbTriplet-1+2\nbTriplet\sumTriplet)$ if all the $3\nbTriplet$ tasks of \setTask are in their median completion interval.\\

Note that for such a schedule to exists, one needs to schedule each task within its median completion time interval. 
Each task from one of the $\mathcal{C}$ sets has a precise median completion time: Any deviation from this time strictly increases the overall deviation. The same holds for the separator tasks, which are scheduled at the same time by all voters. Tasks from $\setTask$, however, are scheduled at different time across voters. Since there are four voters, 
instead of having a single median completion time, a task may have a median completion interval, meaning that half of the voters schedule it to be completed at or after the end of the interval, while the other half schedule it at or before the beginning of the interval. 
This situation occurs here, where tasks from $\setTask$ must be scheduled between the tasks of $\mathcal{C}_2$ and $\mathcal{C}_3$ to minimize the deviation. 
As the separator tasks must be scheduled exactly at their specified completion times in the profile, the only way to schedule the tasks from $\setTask$ without shifting any other task is to form triplets whose total processing time equals the target triplet sum in the \textsc{3-Partition} instance, and to place these triplets between the separator tasks. 
Therefore, the existence of a schedule with a total deviation of $\targetSum$ or less implies the existence of a solution to the \textsc{3-Partition} instance.
One can also simply create a solution for the reduced instance from a solution for the \textsc{3-Partition} instance by putting the tasks corresponding to each triplet between the separator tasks and the rest of the tasks at their median completion time.
Therefore, there exists a schedule in which each task is scheduled at its median completion time (or within its median completion interval) if and only if the corresponding \textsc{3-Partition} instance is a yes-instance. Hence, the problem of finding a schedule that minimizes the total absolute deviation for four voters is NP-hard.
\end{proof}

We shortly note that if one wants to adapt this reduction to the objective function of \citeauthor{kumar2010generalized}~\cite{kumar2010generalized} for weighted candidates, it is only necessary to multiply the offset of the tasks in \setTask by their processing time. The weights do not change the fact that the deviation of a task is minimized when it is scheduled to complete within its median completion interval.

We now move on to the reduction with 3 voters.

\begin{theorem}
The \sigmaDdecprob problem is NP-hard, even when there are only 3 voters.
\label{thm:np_hard}
\end{theorem}

\begin{proof}

\smallskip

We describe a polynomial time reduction from the NP-complete \textsc{3-Partition} problem, which is defined as follows:\\
\noindent\decprob{3-Partition}{A set $\setInt=\{\nameInt_1,\dots,\nameInt_{3\nbTriplet}\}$ of $3\nbTriplet$ integers, an integer \sumTriplet.}{Is there a partition of \setInt\ into \nbTriplet\ triplets $\nameTriplet_1$ to $\nameTriplet_\nbTriplet$ such that $\forall i \in [\nbTriplet]$, the sum of the integers in $\nameTriplet_i$ is exactly \sumTriplet?}\\

\smallskip

This problem is NP-complete even when all integers \( x_i \) are strictly larger than \( \sumTriplet / 4 \) and strictly smaller than \( \sumTriplet / 2 \), and even when the value of \( \sumTriplet \) is polynomially bounded by the number of triplets \( \nbTriplet \)~\cite{garey2002computers}.

We consider an instance of \textsc{3-Partition} that satisfies these assumptions and such that \( \nbTriplet \) is even. Note that the problem remains NP-complete under this assumption. Indeed, given an instance \( I \) with an odd number of triplets, one can construct a new instance \( I' \) by adding three integers with values \( \sumTriplet / 2 - 2 \), \( \sumTriplet / 4 + 1 \), and \( \sumTriplet / 4 + 1 \). In any feasible solution for \( I' \), the integer of value \( \sumTriplet / 2 - 2 \) must necessarily be grouped with the two integers of value \( \sumTriplet / 4 + 1 \), since any other pairing would result in a sum exceeding \( \sumTriplet \). Therefore, \( I' \) is a yes-instance if and only if \( I \) is a yes-instance.

We also assume that \( \sumTriplet \geq 8 \); if not, we multiply the value of \( \sumTriplet \) and all integers by 8 (this clearly does not affect the yes/no status of the instance).\\

Throughout the proof, we provide intuitions in \intuition{italic}. These are not formal descriptions of the reduced instance or the proof, but are intended solely to guide the reader through the reduction.\\

Before describing the reduction, we define 
\[
\factorK = 4\left\lceil\left(\frac{3\nbTriplet^2}{4} + \frac{3\nbTriplet}{2}\right)\right\rceil 
\quad \text{and} \quad 
\sumTripletNew = \sumTriplet \cdot \factorK.
\]

\intuition{The factor \factorK will be used to increase the size of the tasks relative to the original integers. In order to achieve a total deviation of at most \targetSum, a schedule will need to form triplets of tasks summing exactly to \sumTripletNew. Multiplying the values by \factorK significantly penalizes any triplet of tasks that does not sum precisely to \sumTripletNew, introducing an offset of at least \factorK instead of just 1.}\\

\smallskip
We now describe the reduced instance. The set \setTask contains the following subsets of tasks:

\begin{itemize}
    \item A set \(\setIntTask = \{\nameTask_1, \dots, \nameTask_{3\nbTriplet}\}\), such that for all \(i \in [3\nbTriplet]\), the processing time of \(\nameTask_i\) is equal to \(\factorK \cdot \nameInt_i\).
    
    \item A set \(\setLeft = \{\nameLeft_1, \dots, \nameLeft_{\nbTriplet\sumTripletNew}\}\) consisting of \(\nbTriplet\sumTripletNew\) tasks, each with processing time 1.
    
    \item A set \(\setRight = \{\nameRight_1, \dots, \nameRight_{\nbTriplet\sumTripletNew}\}\), also with \(\nbTriplet\sumTripletNew\) tasks of processing time 1.
    
    \item A set \(\setCentral = \{\nameCentral_1, \dots, \nameCentral_{\nbTriplet\sumTripletNew}\}\), with \(\nbTriplet\sumTripletNew\) tasks of processing time 1.
    
    \item A total of \(\nbTriplet + 3\) “separator” sets, denoted \(\setSeparator{0}\) through \(\setSeparator{\nbTriplet+2}\), defined as follows:
    \begin{itemize}
        \item \(\nbTriplet\) “small” separator sets, from \(\setSeparator{1}\) to \(\setSeparator{\nbTriplet}\), each containing 
        \[
        \sizeSeparator = 2\left\lceil \frac{51\nbTriplet^2\sumTriplet}{16} + \frac{\nbTriplet\sumTriplet}{8} \right\rceil
        \]
        tasks of processing time 1.

  \item Two ``large" separator sets \setSeparator{0} and \setSeparator{\nbTriplet+2}, each contains $\sizeLargeSeparator=3(\nbTriplet\sizeSeparator+2\nbTriplet\sumTripletNew)$ tasks of processing time 1,

        
    \end{itemize}
    
    We denote by \(\nameSeparator^i_j\) the \(j^\text{th}\) task from the \(i^\text{th}\) separator set.
\end{itemize}

We now describe the preferences of the three voters. 
 For the sake of readability, 
when tasks from a given set are scheduled consecutively and in non-decreasing order of their indices within the set, we represent them using a single large block. Dashed blocks indicate that only a subset of the set is included (specifically, exactly \(\sumTripletNew\) tasks). Tasks from each set are always scheduled by non-decreasing index. For example, the tasks scheduled by voter 1 between $\nbTriplet\sumTripletNew+\sizeLargeSeparator$ and $\nbTriplet\sumTripletNew+\sizeLargeSeparator+\sumTripletNew$ are the first $\sumTripletNew$ tasks of the set \setLeft.

\begin{figure}[h]
    \begin{tikzpicture}[scale=1.1]

    \hspace{-5em}

    \task{\setIntTask}{1}{1}{2}
    \task{\setSeparator{0}}{1.5}{2.5}{2}
    \taskDashed{\setLeft}{0.5}{3}{2}
    \task{$\setSeparator{1}$}{0.75}{3.75}{2}
    \taskDashed{\setLeft}{0.5}{4.25}{2}
    \task{$\setSeparator{2}$}{0.75}{5}{2}
    \taskDashed{\setLeft}{0.5}{5.5}{2}

    \node at (5.875,2.35) {$\dots$};

    \taskDashed{\setLeft}{0.5}{6.75}{2}
    \task{$\setSeparator{\frac{\nbTriplet}{2}}$}{0.75}{7.5}{2}
    \task{\setCentral}{1}{8.5}{2}    
    \task{$\setSeparator{\frac{\nbTriplet}{2}\!+\!1\!}$}{0.75}{9.25}{2}
    \taskDashed{\setLeft}{0.5}{9.75}{2}
    
    \node at (10.125,2.35) {$\dots$};
    
    \taskDashed{\setLeft}{0.5}{11}{2}    
    \task{$\setSeparator{q}$}{0.75}{11.75}{2}
    \taskDashed{\setLeft}{0.5}{12.25}{2}
    \task{$\setSeparator{\!q\!+\!1\!}$}{0.75}{13}{2}
    \taskDashed{\setLeft}{0.5}{13.5}{2}
    \task{\setSeparator{\nbTriplet+2}}{1.5}{15}{2}
    \task{\setRight}{1}{16}{2}


    \task{\setLeft}{1}{1}{1}
    \task{\setSeparator{0}}{1.5}{2.5}{1}
    \taskDashed{\setRight}{0.5}{3}{1}
    \task{$\setSeparator{1}$}{0.75}{3.75}{1}
    \taskDashed{\setRight}{0.5}{4.25}{1}    
    \task{$\setSeparator{2}$}{0.75}{5}{1}
    \taskDashed{\setRight}{0.5}{5.5}{1}

    \node at (5.875,1.35) {$\dots$};

    \taskDashed{\setRight}{0.5}{6.75}{1}  
    \task{$\setSeparator{\frac{\nbTriplet}{2}}$}{0.75}{7.5}{1}
    \task{\setCentral}{1}{8.5}{1}
    \task{$\setSeparator{\frac{\nbTriplet}{2}\!+\!1\!}$}{0.75}{9.25}{1}
    \taskDashed{\setRight}{0.5}{9.75}{1}

    \node at (10.125,1.35) {$\dots$};

    \taskDashed{\setRight}{0.5}{11}{1}
    \task{$\setSeparator{q}$}{0.75}{11.75}{1}
    \taskDashed{\setRight}{0.5}{12.25}{1}
    \task{$\setSeparator{\!q\!+\!1\!}$}{0.75}{13}{1}
    \taskDashed{\setRight}{0.5}{13.5}{1}
    \task{\setSeparator{\nbTriplet+2}}{1.5}{15}{1}
    \task{\setIntTask}{1}{16}{1}


    \task{\setLeft}{1}{1}{0}
    \task{\setSeparator{0}}{1.5}{2.5}{0}
    \taskDashed{\setCentral}{0.5}{3}{0} 
    \task{$\setSeparator{1}$}{0.75}{3.75}{0}
    \taskDashed{\setCentral}{0.5}{4.25}{0}
    \task{$\setSeparator{2}$}{0.75}{5}{0}  
    \taskDashed{\setCentral}{0.5}{5.5}{0}

    \node at (5.875,0.35) {$\dots$};

    \taskDashed{\setCentral}{0.5}{6.75}{0}  
    \task{$\setSeparator{\frac{\nbTriplet}{2}}$}{0.75}{7.5}{0}
    \task{\setIntTask}{1}{8.5}{0}
    \task{$\setSeparator{\frac{\nbTriplet}{2}\!+\!1\!}$}{0.75}{9.25}{0}
    \taskDashed{\setCentral}{0.5}{9.75}{0}

    \node at (10.125,0.35) {$\dots$};

    \taskDashed{\setCentral}{0.5}{11}{0}
    \task{$\setSeparator{q}$}{0.75}{11.75}{0}
    \taskDashed{\setCentral}{0.5}{12.25}{0}
    \task{$\setSeparator{\!q\!+\!1\!}$}{0.75}{13}{0}
    \taskDashed{\setCentral}{0.5}{13.5}{0}
    \task{\setSeparator{\nbTriplet+2}}{1.5}{15}{0}
    \task{\setRight}{1}{16}{0}

    \draw[->] (-0.2,0) -- (16.1,0);
    
    \draw (0,0.1) -- (0,-0.1) node[below]{0};
    \draw (1,0.1) -- (1,-0.1) node[below]{$\nbTriplet\sumTripletNew$};
    \draw (2.5,0.1) -- (2.5,-0.5) node[below left=-0.1]{$\nbTriplet\sumTripletNew+\sizeLargeSeparator$};
    \draw (3,0.1) -- (3,-1) node[below]{$\nbTriplet\sumTripletNew+\sizeLargeSeparator+\sumTripletNew$};
    
    \end{tikzpicture}

    \caption{High-level view of the preferred schedules of the three voters. }
    \label{fig:highlevel_profilethreevoters}
\end{figure}
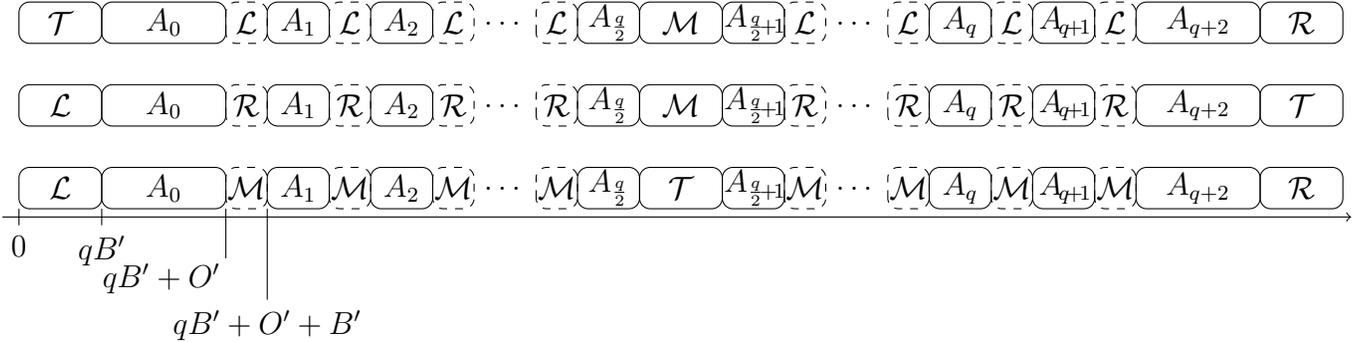

\intuition{The tasks in the separator sets are scheduled at the same positions by the three voters. 
 In any schedule that meets the required objective value \(\targetSum\), these tasks must be placed precisely at those positions, leaving the time intervals in between to be filled by tasks from either \(\setLeft\), \(\setRight\), \(\setCentral\), or \(\setIntTask\). Due to the placement of tasks of \(\setLeft\) and \(\setRight\), we will show that the tasks from these sets must be scheduled either at the very beginning or at the very end of the schedule. This leaves only tasks of \(\setCentral\) and \(\setIntTask\) to fill the intervals between the separators.}\\
 \smallskip

We complete the description of the reduced instance by setting \targetSum.
\begin{equation}
\begin{array}{rl}
    \targetSum = & 2\left(\sum_{i=0}^{\nbTriplet/2-1} \sumTripletNew(\nbTriplet\sumTripletNew+\sizeLargeSeparator+i\sizeSeparator)+\sum_{i=0}^{\nbTriplet/2-1} \sumTripletNew(2\nbTriplet\sumTripletNew+\sizeLargeSeparator+(\nbTriplet/2+1+i)\sizeSeparator)\right) \\
     + & 2\sum_{i=0}^{q/2-1} \sumTripletNew(\sizeSeparator+\nbTriplet\sumTripletNew/2+i\sizeSeparator)\\
     + & \sum_{i=1}^{3\nbTriplet}2(\sizeLargeSeparator+\nbTriplet/2\sizeSeparator+\sumTripletNew\nbTriplet/2+\sumTripletNew\nbTriplet)\\
     + & 51\nbTriplet^2\sumTripletNew/16-\nbTriplet\sumTripletNew/8+3\nbTriplet^2\sizeSeparator/4+3\nbTriplet\sizeSeparator/2
\end{array}
\label{eq:targetSum}
\end{equation}
Note that since that $B'= B\dot K$, and since \factorK is a multiple of 4, \nbTriplet is even (and thus $q^2$ is a multiple of 4) and since \sizeSeparator is a multiple of 2, then \targetSum is an integer. Also, as \factorK, \sizeSeparator and $B'$ are polynomially bounded by $B$ and \nbTriplet, the number of tasks as well as their processing times are all polynomially bounded by \sumTriplet and \nbTriplet. It follows that \targetSum is also polynomially bounded: The reduced instance can be obtained in polynomial time.

\smallskip 

\intuition{In Equation \ref{eq:targetSum}, the value of \targetSum is spread across four lines. The first three lines of Equation~\ref{eq:targetSum} form a lower bound on the deviation of any solution. The last line is an upper bound on the additional cost of a solution corresponding to a {\sc 3-Partition}. Any solution which does not split the tasks of \setIntTask into triplets of adequate processing times will have a larger additional cost.}

\smallskip

Throughout the proof, we say that a task from a separator set is on time in schedule \sched if it completes at the same time in \sched as in the preferences of the three voters, since they agree on the completion times of such tasks. Similarly, a task from a separator set is said to be early (resp. late) in \sched if it completes earlier (resp. later) in \sched than in the voters' preferences.  We begin by proving a couple of lemmas.

\begin{lemma}
Given a schedule \sched, one can construct a schedule $\sched'$ with a total sum of absolute deviations no greater than that of \sched, and such that: 
\begin{itemize}
    \item In $\sched'$, for every separator set \setSeparator{i}, all tasks from the set are scheduled consecutively and in non-decreasing order of their indices.
    \item in $\sched'$, each task from a separator set completes at most $\sumTripletNew/4$ units of time away from its completion time in the voters' preferences.
\end{itemize} 
\label{lem:separator}
\end{lemma}

\begin{proof}
First, tasks within any given separator set can be scheduled in non-decreasing order of their indices. Indeed, one can take any schedule that does not respect this ordering and iteratively swap tasks within the same separator set until the correct order is achieved — without increasing the total deviation. A similar exchange argument applies to the relative ordering of different separator sets: For each integer $i \in [q+1]$, tasks from \setSeparator{i} precede those from \setSeparator{i+1} in all voters' preferences. Therefore, in any optimal solution, we can assume without loss of generality that tasks from \setSeparator{i} are scheduled before those from \setSeparator{i+1}.


Secondly, consider a schedule $\sched^*$ that minimizes the total sum of absolute deviations and satisfies the ordering properties established above. Suppose that in $\sched^*$, the tasks of a given separator set $\setSeparator{i}$ are not scheduled as a contiguous block. That is, there exists a smallest index $j$ such that at least one task is scheduled between $\nameSeparator_j^i$ and $\nameSeparator_{j+1}^i$.

Since these two tasks appear consecutively in the preferences of all voters, their separation in $\sched^*$ implies a misalignment: Specifically, either $\nameSeparator_j^i$ completes earlier than in the voters' preferences (i.e., it is \emph{early}), or $\nameSeparator_{j+1}^i$ completes later (i.e., it is \emph{late}), or both. We distinguish the following three cases:
\begin{enumerate}
    \item $\nameSeparator_j^i$ is early, and $\nameSeparator_{j+1}^i$ is either early or on time;
    \item $\nameSeparator_j^i$ is either late or on time, and $\nameSeparator_{j+1}^i$ is late;
    \item $\nameSeparator_j^i$ is early, and $\nameSeparator_{j+1}^i$ is late.
\end{enumerate}


\paragraph{Case (1)} Let $t$ be the task scheduled immediately after $\nameSeparator_j^i$ in $\sched^*$. Since $j$ is the first index such that $\nameSeparator_j^i$ and $\nameSeparator_{j+1}^i$ are not scheduled consecutively, all tasks $\nameSeparator_1^i$ through $\nameSeparator_{j-1}^i$ must be scheduled directly before $\nameSeparator_j^i$ in $\sched^*$. Therefore, these tasks are also early by the same amount of time.

Let $\sched$ be the schedule obtained from $\sched^*$ by moving task $t$ to the position just before $\nameSeparator_1^i$. The total sum of absolute deviations between $\sched$ and $\sched^*$ remains unchanged for all tasks except $t$ and the tasks $\nameSeparator_1^i$ through $\nameSeparator_j^i$.

In $\sched$, the deviation of each task $\nameSeparator_1^i, \dots, \nameSeparator_j^i$ is reduced by $3$ times the processing time of $t$, i.e., by at least $3j$ in total, since there are 3 voters and $t$ has processing time at least $1$ and is moved $j$ positions. As $\nameSeparator_{j+1}^i$ is assumed to be early or on time, the tasks $\nameSeparator_1^i$ through $\nameSeparator_j^i$ in $\sched$ are either still early (but less so) or on time.

Task $t$, now completing $j$ time units earlier in $\sched$ compared to $\sched^*$, may see its deviation increase by at most $j$ per voter, i.e., at most $3j$ in total. Consequently, the total deviation in $\sched$ is no greater than in $\sched^*$, which implies that $\sched$ is also optimal.

We can repeat this swapping process until there are no tasks scheduled between $\nameSeparator_j^i$ and $\nameSeparator_{j+1}^i$, ensuring that the separator set becomes a contiguous block.


\paragraph{Case (2)} This case is similar to the first. Let $t$ be the task scheduled immediately before $\nameSeparator_{j+1}^i$ in $\sched^*$. We construct a new schedule by swapping $t$ and $\nameSeparator_{j+1}^i$.

In the new schedule, the deviation of $\nameSeparator_{j+1}^i$ is reduced by at least $3$, since it is moved one position earlier and is currently late. The deviation of task $t$ may increase by at most $3$, as it is moved one position later and each voter's deviation can increase by at most $1$ unit.

Hence, the total sum of deviations does not increase, and the resulting schedule remains optimal. We repeat this process until no task remains between $\nameSeparator_j^i$ and $\nameSeparator_{j+1}^i$, thereby bringing the two tasks adjacent in the schedule.


\paragraph{Case (3)} This case is slightly more intricate. Suppose first that there are at least two tasks scheduled between $\nameSeparator_j^i$ and $\nameSeparator_{j+1}^i$ in $\sched^*$. In that situation, we can either delay the block of tasks $\nameSeparator_1^i, \dots, \nameSeparator_j^i$, or advance $\nameSeparator_{j+1}^i$. After the swap, $\nameSeparator_j^i$ remains early or $\nameSeparator_{j+1}^i$ remains late, allowing us to apply the same reasoning as in Cases (1) and (2): Each unit of shift reduces the deviation of one or more tasks by 3, while any potential increase in deviation is at most 3. Therefore, the swap does not increase the total deviation and can be repeated until the tasks become adjacent.

The subtlety arises when exactly one task $t$ is scheduled between $\nameSeparator_j^i$ and $\nameSeparator_{j+1}^i$. In this case, swapping $t$ with either $\nameSeparator_j^i$ or $\nameSeparator_{j+1}^i$ might cause a task that was early to become late, or vice versa, potentially neutralizing the benefit of the move.

To analyze this scenario, observe that in all voters' preferences, $\nameSeparator_j^i$ completes exactly one unit of time before $\nameSeparator_{j+1}^i$. Since $t$ is the only task in between, and $\nameSeparator_j^i$ is early while $\nameSeparator_{j+1}^i$ is late, there must exist a strictly positive integer $d \leq \procTime_t$ such that:
\begin{itemize}
    \item each task $\nameSeparator_1^i, \dots, \nameSeparator_j^i$ is early by $d$ time units, and
    \item $\nameSeparator_{j+1}^i$ is late by $\procTime_t - d$ time units.
\end{itemize}

We then distinguish three subcases:
\begin{enumerate}
    \item[(3.1)] $d > \procTime_t - d$,
    \item[(3.2)] $d < \procTime_t - d$,
    \item[(3.3)] $d = \procTime_t - d$.
\end{enumerate}

These subcases correspond to different balances between the early and late deviations, which determine whether the swap is beneficial. Figure~\ref{fig:prop_separator_case_3} illustrates Case (3.2), where $\procTime_t - d > d$. By adjusting the value of $d$ in the figure, one can visualize how the swap operation affects the total deviation in each subcase.

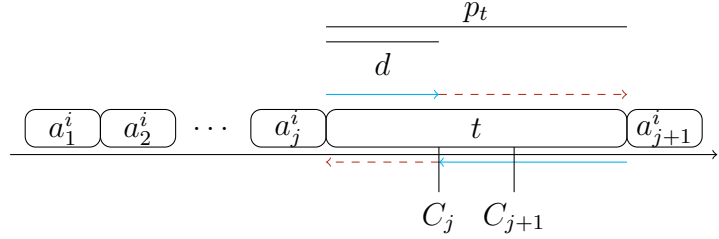
\begin{figure}[h]
    \centering
    \begin{tikzpicture}


        \task{$\nameSeparator_1^i$}{1}{1}{0}
        \task{$\nameSeparator_2^i$}{1}{2}{0}
        \node at (2.5,0.35) {$\dots$};
        \task{$\nameSeparator_j^i$}{1}{4}{0}
        \task{$t$}{4}{8}{0}
        \task{$\nameSeparator_{j+1}^i$}{1}{9}{0}

        \draw[->] (-0.2,0) -- (9.2,0);

        \draw (4,1.7) -- (8,1.7); 
        \node at (6,1.9) {$\procTime_t$};
        \draw (4,1.5) -- (5.5,1.5); 
        \node at (4.75,1.2) {$d$};

        \draw [->, cyan] (4,0.8) -- (5.5,0.8);
        \draw [->, dashed, BrickRed] (5.5,0.8) -- (8,0.8);

        \draw [->, cyan] (8,-0.1) -- (5.5,-0.1);
        \draw [->, dashed, BrickRed] (5.5,-0.1) -- (4,-0.1);
        
        \draw (5.5,0.1) -- (5.5,-0.5) node[below]{$C_j$};
        \draw (6.5,0.1) -- (6.5,-0.5) node[below]{$C_{j+1}$};
    \end{tikzpicture}
    \caption{Illustration of Case (3). $C_j$ (resp. $C_{j+1}$) denotes the completion time of $\nameSeparator_j^i$ (resp. $\nameSeparator_{j+1}^i$) in the voters' preferences. A light blue arrow indicates a reduction in deviation when the corresponding task is moved in the direction of the arrow, while a dashed dark red arrow indicates an increase in deviation.}
    \label{fig:prop_separator_case_3}
\end{figure}


Now observe that if we swap the block of tasks $\nameSeparator_1^i$ through $\nameSeparator_j^i$ with task $t$, the deviation for each task in the block is reduced by $3 \cdot d$ and increased by $3 \cdot (\procTime_t - d)$.

This yields the following conclusions:
\begin{itemize}
    \item \textbf{Case (3.1):} $d > \procTime_t - d$. The deviation of each task in the block decreases more than it increases, leading to a net reduction of at least $3j$ for the block. Task $t$ may experience an increase in deviation of at most $3j$, since it is moved by $j$ time units and there are 3 voters. Therefore, the total deviation does not increase, and the swap yields another optimal schedule.
    
    \item \textbf{Case (3.2):} $d < \procTime_t - d$. We instead swap $t$ with $\nameSeparator_{j+1}^i$. By similar reasoning, the total deviation does not increase.
    
    \item \textbf{Case (3.3):} $d = \procTime_t - d$. Swapping $t$ with either the block $\nameSeparator_1^i, \dots, \nameSeparator_j^i$, or with $\nameSeparator_{j+1}^i$, results in no net change in deviation for the involved tasks.
\end{itemize}

Finally, note that the completion times of $\nameSeparator_j^i$ and $\nameSeparator_{j+1}^i$ in the profile occur during the execution of $t$ in $\sched^*$. This means that no voter schedules $t$ to complete at exactly the same time as in $\sched^*$.

After performing the swaps described in the corresponding case, the first index $j'$, if it exists, such that there is at least one task scheduled between $\nameSeparator_{j'}^i$ and $\nameSeparator_{j'+1}^i$, must satisfy $j' > j$. This implies that by iteratively applying such swaps, we eventually obtain a schedule in which all tasks from $\setSeparator{i}$ appear consecutively and are ordered by non-decreasing index.

Now, observe that once the tasks of a separator set are scheduled as a contiguous block and in the correct order, they must all be either early by the same amount of time, late by the same amount, or on time. Suppose they are all early by $x$ time units. Consider the task $t$ scheduled immediately after the last task of the separator set. If the processing time of $t$ is strictly less than $2x$, then moving $t$ before the separator block would reduce the total deviation: Each task in the separator set would see its deviation increase by at most $p_t - x$, while $t$ would see its deviation decrease by $x$ for each voter. Thus, to prevent such an improving swap, $t$ must have processing time at least $2x + 1$.

Since the largest task in the instance has processing time at most $\sumTripletNew/2$, we conclude that $x \leq \sumTripletNew/4$. A symmetric argument shows that the tasks of the separator set cannot be late by more than $\sumTripletNew/4$ units of time either.

Therefore, starting from a given schedule $\sched^*$, it is possible to construct a schedule $\sched$ in which:
\begin{itemize}
    \item tasks from each separator set are scheduled consecutively,
    \item they are ordered by non-decreasing index within each set,
    \item each task from a separator set is either early by at most $\sumTripletNew/4$ units of time, or on time, or late by at most $\sumTripletNew/4$ units of time, and
    \item the total sum of absolute deviations in $\sched$ is no greater than in $\sched^*$.
\end{itemize}
\end{proof}

From now on, when considering an optimal solution $\sched^*$, we will systematically consider that $\sched^*$ is such that tasks from separator sets are scheduled in blocks and ``close" to, i.e., less than $\sumTripletNew/4$ units of time away from, their positions in the preferences of the voters. Using Proposition~\ref{prop:unanimity_egalite_sigmad}, we get the following observation.



\begin{obs}
One can transform any given solution $\sched$ into a schedule $\sched'$ with no greater total deviation, such that:
\begin{itemize}
    \item the tasks in the sets $\setLeft$, $\setRight$, and $\setCentral$ are ordered by non-decreasing index within each set,
    \item all tasks from $\setLeft$ are scheduled before all tasks from $\setRight$, and
    \item all tasks from $\setCentral$ are scheduled between the separator sets $\setSeparator{0}$ and $\setSeparator{\nbTriplet+2}$.
    \item  The total deviation in $\sched'$ is not larger than the total deviation in \sched.
\end{itemize}

\label{obs:unanimity}
\end{obs}

From now on, when considering an optimal solution $\sched^*$, we will systematically consider that $\sched^*$ is such that tasks from the sets \setLeft, \setRight and \setCentral are scheduled by non-decreasing index and such that all tasks from \setLeft are scheduled before all tasks from \setRight and tasks of \setCentral are scheduled after \setSeparator{0} and before \setSeparator{\nbTriplet+2}.

\begin{lemma}
Given any schedule $\sched$, there exists a schedule $\sched'$ in which all tasks from the separator sets $\setSeparator{0}$ and $\setSeparator{\nbTriplet+2}$ complete at exactly the same times as in the voters' preferences, and such that the total sum of absolute deviations in $\sched'$ is no greater than in $\sched$.
\label{lem:large_separator}
\end{lemma}

\begin{proof}
By Lemma~\ref{lem:separator}, 
tasks from \setSeparator{1} start at the earliest at time $\nbTriplet\sumTripletNew+\sizeLargeSeparator+3\sumTripletNew/4$. This means that the tasks scheduled before \setSeparator{1} have a total processing time of at least $\nbTriplet\sumTripletNew+\sizeLargeSeparator+3\sumTripletNew/4$. If we remove the total processing time of tasks from \setSeparator{0}, there remains at least  $\nbTriplet\sumTripletNew+3\sumTripletNew/4$ units of time occupied by tasks outside of \setSeparator{0}. As the tasks from \setIntTask have a sum of processing times of $\nbTriplet\sumTripletNew$ and the rest of the tasks are unit tasks, it means that there are at least $3\sumTripletNew/4$ unit tasks in the considered time interval. Let us  consider a schedule \sched in which the tasks of \setSeparator{0} are not on time. There are two cases:
\begin{enumerate}

    \item[\emph{Case (1):}] \emph{The tasks of $\setSeparator{0}$ are scheduled late.}

In this case, there must be at least one unit task scheduled before $\setSeparator{0}$, since the tasks from $\setIntTask$ occupy at most $\nbTriplet\sumTripletNew$ units of time. Let $j$ be the last unit task scheduled before $\setSeparator{0}$. Consider the schedule obtained by moving task $j$ immediately after the tasks of $\setSeparator{0}$.

This move may increase the deviation of task $j$ by at most $\nbTriplet\sumTripletNew + \sizeLargeSeparator + \tfrac{3}{4} \sumTripletNew$, as it is delayed by at most that many units. Additionally, all tasks from $\setIntTask$ that were scheduled between $j$ and $\setSeparator{0}$ are now completed one unit earlier.
Since this is unsatisfactory only for the last two voters, their deviation increases by at most 2 units per task. There are at most $3\nbTriplet$ such tasks, leading to a total increase in deviation of at most $2 \cdot 3\nbTriplet = 6\nbTriplet$.

However, each task in $\setSeparator{0}$ is now completed one unit earlier, reducing its deviation by 1 per voter, i.e., 3 per task. Since there are $\sizeLargeSeparator$ such tasks, the total decrease in deviation is $3 \cdot \sizeLargeSeparator$, which is strictly greater than the total increase from moving $j$. Hence, the swap strictly reduces the total deviation.

   


\item[\emph{Case (2):}] \emph{The tasks of $\setSeparator{0}$ are scheduled early.}

This case is illustrated in Figure~\ref{fig:prop_large_separator_case_2}. Let us consider the first task from $\setLeft$ that is scheduled after the tasks of $\setSeparator{0}$, and denote it by $\nameLeft$. Task $\nameLeft$ must be scheduled before the tasks from $\setRight$ and from $\setSeparator{\nbTriplet+2}$.

We construct a new schedule $\sched'$ from $\sched$ by moving $\nameLeft$ to a position just before the tasks of $\setSeparator{0}$. This change reduces the deviation of $\nameLeft$ by at least $\sizeSeparator$ plus the number of time units between the end of $\setSeparator{0}$ and the original completion time of $\nameLeft$.

This time gap is at most $\nbTriplet\sizeSeparator + 2\nbTriplet\sumTripletNew$, as there are at most $\nbTriplet$ small separator sets and at most $2\nbTriplet$ tasks from $\setCentral$ and $\setIntTask$. Note that this value is strictly less than $\sizeLargeSeparator$.

Each task scheduled between the end of $\setSeparator{0}$ and the original completion time of $\nameLeft$ is now delayed by one unit of time. There are at most $\sizeLargeSeparator$ such tasks, each incurring a deviation increase of at most 3 units (one per voter), for a total increase of at most $3(\nbTriplet\sizeSeparator + 2\nbTriplet\sumTripletNew) = \sizeLargeSeparator$.

Since each task in $\setSeparator{0}$ now completes one unit later, i.e., closer to its target completion time, the deviation is reduced by 3 per task, yielding a total reduction of $3\sizeLargeSeparator$.

Therefore, the total deviation strictly decreases after performing the swap.


    \begin{figure}[h]
    \hspace{-5em}
    \begin{tikzpicture}


        \task{$\nameSeparator_1^0$}{1}{3}{2}
        \task{$\nameSeparator_2^0$}{1}{4}{2}
        \node at (6,2.35) {$\dots$};
        \task{$\nameSeparator_{\sizeLargeSeparator}^0$}{1}{9}{2}
        \taskDashed{}{4}{13}{2}
        \task{\nameLeft}{1}{14}{2}

        \draw (3,1.45) -- (10,1.45);
        \draw (10,1.45) -- (14,1.45);
        \draw[cyan] (13.5,2) -- (13.5,1.5);
        \draw[cyan] (13.5,1.5) -- (2.5,1.5);
        \draw[->,cyan] (2.5,1.5) -- (2.5,1);
        \draw (3,1.4) -- (3,1.6);
        \draw (10,1.4) -- (10,1.6);
        \node at (6.5,1.2) {$\sizeLargeSeparator$};
        
        \draw (14,1.4) -- (14,1.6);
        \node at (12,1.2) {$\leq\sizeLargeSeparator$};

        \task{\nameLeft}{1}{3}{0}
        \task{$\nameSeparator_1^0$}{1}{4}{0}
        \task{$\nameSeparator_2^0$}{1}{5}{0}
        \node at (7,0.35) {$\dots$};
        \task{$\nameSeparator_{\sizeLargeSeparator}^0$}{1}{10}{0}
        \taskDashed{}{4}{14}{0}

        \draw[->] (-0.2,0) -- (15.2,0);

        \draw [->, cyan] (3,-0.15) -- (4,-0.15) node[below left=0.5em] {\color{cyan} -3};
        \draw [->, cyan] (4,-0.15) -- (5,-0.15) node[below left=0.5em] {\color{cyan} -3};
        \draw [->, cyan] (5,-0.15) -- (6,-0.15) node[below left=0.5em] {\color{cyan} -3};
        \draw [->, cyan] (9,-0.15) -- (10,-0.15) node[below left=0.5em] {\color{cyan} -3};
        
        \draw [->, dashed, BrickRed] (10,-0.15) -- (11,-0.15) node[below left=0.5em] {\color{BrickRed} $\leq 3$};
        \draw [->, dashed, BrickRed] (11,-0.15) -- (12,-0.15) node[below left=0.5em] {\color{BrickRed} $\leq 3$};
        \draw [->, dashed, BrickRed] (12,-0.15) -- (13,-0.15) node[below left=0.5em] {\color{BrickRed} $\leq 3$};
        \draw [->, dashed, BrickRed] (13,-0.15) -- (14,-0.15) node[below left=0.5em] {\color{BrickRed} $\leq 3$};
        
        \draw (4,0.1) -- (4,-0.1) node[below]{$C_0$};
    \end{tikzpicture}
    \caption{Representation of Case (2). $C_0$ denotes the completion time of $\nameSeparator_1^0$ in the voter's preferences.  A light blue (resp. dashed dark red) arrow indicates that the deviation decreases (resp. increases) when the task is moved in the direction of the arrow, and the light blue (resp. dark red) value below it indicates the corresponding magnitude of change.}

    \label{fig:prop_large_separator_case_2}
\end{figure}
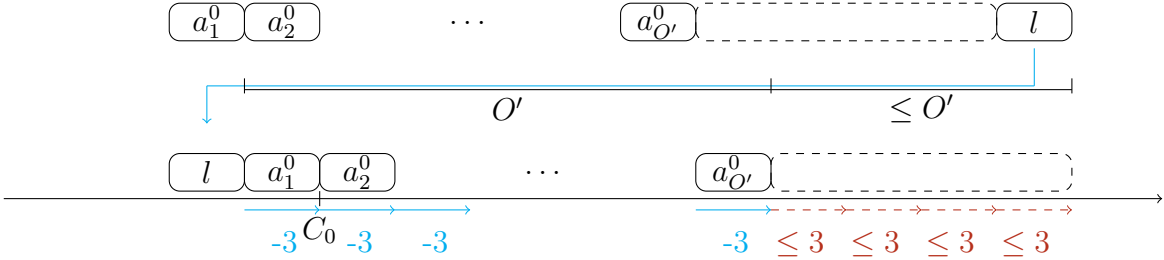
\end{enumerate}
The proof is symmetrical for \setSeparator{\nbTriplet+2}.
\end{proof}

From now on, when considering an optimal schedule $\sched^*$, we will consider that $\sched^*$ is such that tasks from \setSeparator{0} and \setSeparator{\nbTriplet+2} are on time.

\begin{lemma}
From a given optimal schedule $\sched^*$, one can obtain an optimal schedule $\sched'$ such that tasks from \setLeft are scheduled at the beginning of $\sched$ and tasks from \setRight are scheduled at the end of \sched. 
\label{lem:left_right}
\end{lemma}

\begin{proof}
Let us consider an optimal solution $\sched^*$ which fulfills all conditions mentioned above. 


Let us consider that in $\sched^*$ at least one task from \setLeft is scheduled after \setSeparator{0}. As no task except tasks from \setIntTask and \setLeft can be scheduled before \setSeparator{0} and as tasks from \setSeparator{0} are on time, it means that there is at least one task from \setIntTask scheduled before \setSeparator{0}. As tasks from \setSeparator{0} are on time and as the total processing time of tasks in \setLeft on the one hand, and of \setIntTask on the other hand, is precisely $\nbTriplet\sumTripletNew$, the tasks from \setIntTask scheduled before \setSeparator{0} have a total processing time equal to the number of tasks from \setLeft that are scheduled after \setSeparator{0}.

First let us assume that the tasks scheduled before \setSeparator{0} are such that all tasks from \setLeft are scheduled first and in non-decreasing order of their indices, as if a task from \setLeft is scheduled after either a task from \setIntTask or a task from \setLeft with a smaller index, then we can swap them without increasing the deviation.

We will perform the following swap from $\sched^*$ to obtain a new schedule \sched. Let $j$ be the last task from \setIntTask scheduled before \setSeparator{0}. We move $j$ just after \setSeparator{0}. We also move the first $\procTime_j$ tasks from \setLeft which were scheduled after  \setSeparator{0} and which are now scheduled before \setSeparator{0} in the time interval which was occupied by $j$ in $\sched^*$. We will now show that such a swap lowers the deviation.

First observe that tasks from \setLeft are scheduled before tasks from \setRight and \setSeparator{\nbTriplet+2}. This means that there are at most $\nbTriplet\sizeSeparator+2\nbTriplet\sumTripletNew$ tasks scheduled in between \setSeparator{0} and the last task from \setLeft. This implies that there are at most $\nbTriplet\sizeSeparator+2\nbTriplet\sumTripletNew$ tasks which are delayed by at most $\procTime_j$ units of time, causing a potential increase in the deviation of at most $3\procTime_j \cdot (\nbTriplet\sizeSeparator+2\nbTriplet\sumTripletNew)$. 
On the other hand, $\procTime_j$ tasks from \setLeft have a deviation reduced by at least \sizeLargeSeparator each, which lowers the deviation by at least $\procTime_j \cdot\sizeLargeSeparator$, which is precisely $3\procTime_j \cdot (\nbTriplet\sizeSeparator+2\nbTriplet\sumTripletNew)$.

The proof for \setRight and \setSeparator{\nbTriplet+2} is symmetrical.
\end{proof}

Lemmas~\ref{lem:separator}, \ref{lem:large_separator}, and~\ref{lem:left_right} and Observation~\ref{obs:unanimity} imply the existence of an optimal solution with the following structure: Tasks from \setLeft, followed by tasks of \setSeparator{0}, then tasks from standard separator sets (\setSeparator{1}, \dots , \setSeparator{q+1}), which are scheduled close to their completion times in the preferences. Between them, there are tasks of \setIntTask. The schedule ends with tasks from \setSeparator{\nbTriplet+2} followed by tasks from \setRight. The shape of such a schedule is represented in Figure~\ref{fig:optsol}.

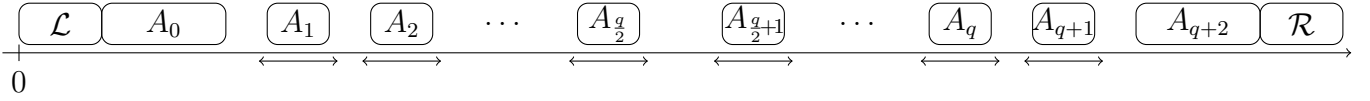
\begin{figure}[h]
    
    \begin{tikzpicture}[scale=1.1]

\hspace{-5em}

    \task{\setLeft}{1}{1}{0}
    \task{\setSeparator{0}}{1.5}{2.5}{0}
    \task{$\setSeparator{1}$}{0.75}{3.75}{0}
    \task{$\setSeparator{2}$}{0.75}{5}{0}

    \node at (5.875,0.35) {$\dots$};

    \task{$\setSeparator{\frac{\nbTriplet}{2}}$}{0.75}{7.5}{0}
    \task{$\setSeparator{\frac{\nbTriplet}{2}\!+\!1\!}$}{0.75}{9.25}{0}

    \node at (10.1625,0.35) {$\dots$};

    \task{$\setSeparator{q}$}{0.75}{11.75}{0}
    \task{$\setSeparator{q+1}$}{0.75}{13}{0}
    \task{\setSeparator{\nbTriplet+2}}{1.5}{15}{0}
    \task{\setRight}{1}{16}{0}

    \draw[->] (-0.2,0) -- (16.1,0);
    
    \draw (0,0.1) -- (0,-0.1) node[below]{0};
    \draw[<->] (2.9,-0.1) --(3.85,-0.1);
    \draw[<->] (4.15,-0.1) --(5.1,-0.1);
    \draw[<->] (6.65,-0.1) --(7.6,-0.1);
    \draw[<->] (8.4,-0.1) --(9.35,-0.1);
    \draw[<->] (10.9,-0.1) --(11.85,-0.1);
    \draw[<->] (12.15,-0.1) --(13.1,-0.1);
    
    \end{tikzpicture}

    \caption{Structure of an existing optimal solution. Arrows indicate that the separator sets, from \setSeparator{1} to \setSeparator{\nbTriplet+1}, complete close to their completion times in the voters' preferences.}
    \label{fig:optsol}
\end{figure}


We now look at the deviation caused by tasks which position is fixed in such a solution. First note that tasks from \setSeparator{0} and \setSeparator{\nbTriplet+2} are scheduled precisely at the same time than in the preferences: Their deviation is then always 0 in an optimal solution. Tasks from \setRight and \setLeft are scheduled at their median completion time. Task $\nameLeft_1$ completes in the preference of voter one at time $\nbTriplet\sumTripletNew+\sizeLargeSeparator+1$ and at time $1$ in an optimal schedule: This is a deviation of $\nbTriplet\sumTripletNew+\sizeLargeSeparator$. 
Note that the deviation is the same for all tasks up to $\nameLeft_{\sumTripletNew}$; beyond this point, an additional offset of \sizeSeparator is added each time the tasks pass a separator set. 
Overall, the deviation caused by the tasks of \setLeft when scheduled at the beginning of an optimal solution, denoted by $D_{\setLeft}$, is:

$$
D_{\setLeft}=\sum_{i=0}^{\nbTriplet/2-1} \sumTripletNew(\nbTriplet\sumTripletNew+\sizeLargeSeparator+i\sizeSeparator)+\sum_{i=0}^{\nbTriplet/2-1} \sumTripletNew(2\nbTriplet\sumTripletNew+\sizeLargeSeparator+(\nbTriplet/2+1+i)\sizeSeparator)
$$



As the sets \setLeft and \setRight are constructed and scheduled symmetrically, the corresponding cost, $D_{\setRight}$, is therefore equal to $D_{\setLeft}$. 

Now, observe that tasks achieve a minimum deviation when they are scheduled to complete exactly at their median completion time. 
We next derive a lower bound on the deviation of any schedule by computing the deviation of an (infeasible) schedule in which every task completes exactly at its median completion time. 
Such a schedule is illustrated in Figure~\ref{fig:nfsol}.

\begin{figure}[h]
    
    \begin{tikzpicture}[scale=1.1]

\hspace{-5em}

    \task{\setLeft}{1}{1}{0}
    \task{\setSeparator{0}}{1.5}{2.5}{0}
    \task{$\setSeparator{1}$}{0.75}{3.75}{0}
    \task{$\setSeparator{2}$}{0.75}{5}{0}

    \node at (5.875,0.35) {$\dots$};

    \task{$\setSeparator{\frac{\nbTriplet}{2}}$}{0.75}{7.5}{0}
    \task{\setCentral}{1}{8.5}{-0.25}
    \task{\setIntTask}{1}{8.5}{0.25} 
    \task{$\setSeparator{\frac{\nbTriplet}{2}\!+\!1\!}$}{0.75}{9.25}{0}

    \node at (10.125,0.35) {$\dots$};

    \task{$\setSeparator{q}$}{0.75}{11.75}{0}
    \task{$\setSeparator{q+1}$}{0.75}{13}{0}
    \task{\setSeparator{\nbTriplet+2}}{1.5}{15}{0}
    \task{\setRight}{1}{16}{0}

    \draw[->] (-0.2,-0.25) -- (16.1,-0.25);
    
    \draw (0,-0.15) -- (0,-0.35) node[below]{0};
    
    \end{tikzpicture}

    \caption{Infeasible ideal schedule. }
    \label{fig:nfsol}
\end{figure}
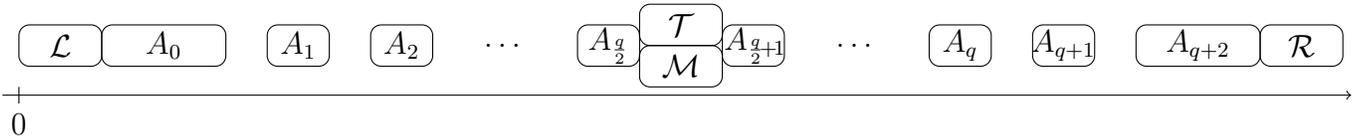

In such a schedule, every task from a separator set completes at the same time as in the voters' preferences; their deviation is thus 0. 
This leaves the tasks from \setCentral and \setIntTask, which are scheduled between  \setSeparator{\nbTriplet/2} and \setSeparator{\nbTriplet/2+1}. 
We begin by computing the cost of scheduling the tasks from \setCentral at their median completion time, denoted by $D_{\setCentral}$:

$$
D_{\setCentral}=2\sum_{i=0}^{q/2-1} \sumTripletNew(\sizeSeparator+\nbTriplet\sumTripletNew/2+i\sizeSeparator)
$$

\intuition{
We provide the intuition for the first half of the schedule, noting that the same reasoning applies symmetrically to the second half. 
We decompose the formula by first examining each task individually. The distance between the median completion time of a task and its completion time in the preferences of voter~3 can be expressed in terms of the number of separators and additional time units. 
Consider a schedule in which all tasks from \setCentral are moved before  \setSeparator{\nbTriplet/2} (temporarily ignoring other separators).  
Each task is then shifted by exactly \sizeSeparator plus $\nbTriplet\sumTripletNew/2$  time units. 
Intuitively, if such an operation is performed iteratively, starting with the last task, each task from \setCentral must be moved before every other task in the set. 
Now, in such a schedule, insert a separator set after each block of \sumTripletNew tasks,  so that the tasks are divided into $\nbTriplet/2$ sets of size \sumTripletNew. 
The last set has no additional offset, the second-to-last has an offset of  \sizeSeparator, and so on, up to the first set, which has an offset of $\sizeSeparator(\nbTriplet/2 - 1)$, thereby yielding the formula above.
} 

\smallskip 

We now give the cost of scheduling tasks from \setIntTask at their median completion time, denoted by $D_{\setIntTask}$:

$$
D_{\setIntTask}=\sum_{i=1}^{3\nbTriplet}2(\sizeLargeSeparator+\nbTriplet/2(\sizeSeparator+\sumTripletNew)+\nbTriplet\sumTripletNew)
$$

\intuition{As the tasks from \setIntTask are always scheduled in the same order, the difference between their completion times in the schedules of the third and first voters remains constant across all tasks.  
Specifically, it is equal to $3\nbTriplet$ times the difference between the starting times of the tasks from \setIntTask in the voters' preferences.}

\smallskip 
\begin{figure}[h]
    
    \begin{tikzpicture}[scale=1.1]

\hspace{-5em}

    \task{\setSeparator{0}}{1.5}{2.5}{0}
    \task{$\setSeparator{1}$}{0.75}{3.75}{0}
    \task{$\setSeparator{2}$}{0.75}{5}{0}

    \node at (5.875,0.35) {$\dots$};

    \task{$\setSeparator{\frac{\nbTriplet}{2}}$}{0.75}{7.75}{0}
    \task{\setIntTask}{1}{8.75}{0}
    \task{$\setSeparator{\frac{\nbTriplet}{2}\!+\!1\!}$}{0.75}{9.5}{0}

    \node at (10.125,0.35) {$\dots$};

    \task{$\setSeparator{q}$}{0.75}{11.75}{0}
    \task{$\setSeparator{q+1}$}{0.75}{13}{0}
    \task{\setSeparator{\nbTriplet+2}}{1.5}{15}{0}

    \draw[->] (-0.2,-0.5) -- (16.1,-0.5);
    
    \draw (0,-0.4) -- (0,-0.6) node[below]{0};

    \draw[->, BrickRed] (0,0.7) -- (7.75,0.7); 

    \node at (0.5,1.1) {$\nbTriplet\sumTripletNew$};
    \draw (1,0.8) -- (1,0.6);
    \node at (1.75,1.1) {$\sizeLargeSeparator$};
    \draw (2.5,0.8) -- (2.5,0.6);
    \node at (5,1.1) {$\nbTriplet/2 (\sizeSeparator+\sumTripletNew)$};
    
    \draw[->, BrickRed] (15,0) -- (7.75,0); 

    \node at (8.25,-0.3) {$\nbTriplet\sumTripletNew$};
    \draw (8.75,0.1) -- (8.75,-0.1);
    \node at (11.25,-0.3) {$\nbTriplet/2 (\sizeSeparator+\sumTripletNew)$};
    \draw (13.5,0.1) -- (13.5,-0.1);
    \node at (14.25,-0.3) {$\sizeLargeSeparator$};

    \end{tikzpicture}

    \caption{Representation of the deviation of integer tasks when scheduled at their median completion time.}
    \label{fig:median_int_task}
\end{figure}
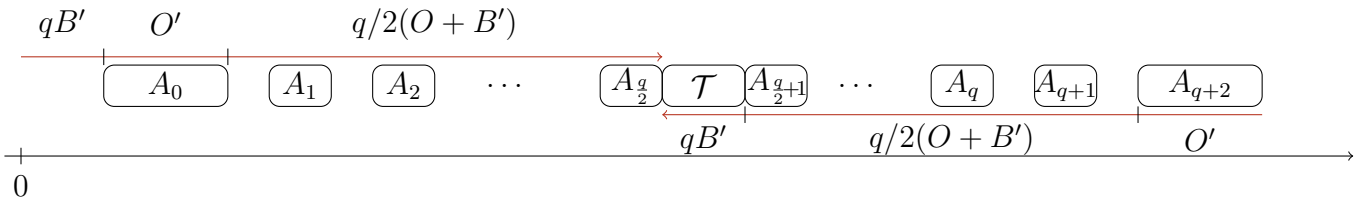

Let us denote by $D_{NF}$  the cost of the infeasible schedule in which each task completes exactly at its median completion time. We have: $D_{NF} = D_{\setLeft} + D_{\setRight} + D_{\setCentral} + D_{\setIntTask}$


Note that if a task from \setCentral or \setIntTask deviates from its median completion time, then the additional deviation (compared to being scheduled at its median completion time) is equal to one times the difference, as one of the three voters becomes more satisfied while two become dissatisfied by such a change. 
This property holds for \setCentral because tasks from \setCentral are scheduled in non-decreasing order of their indices in an optimal solution, and since the tasks from \setLeft, \setRight, \setSeparator{0}, and \setSeparator{\nbTriplet+2} are fixed in position, no task from \setCentral can be scheduled beyond its completion time in the preference of voter~3 (earlier for the first half or later for the second half). 
On the other hand, if a task from a separator set deviates from its median completion time, the additional deviation is three times the difference.

\bigskip

We will now show that the reduced instance of \sigmaDdecprob is a yes-instance if and only if the \textsc{3-Partition} instance is a yes-instance.

\smallskip

\begin{itemize}
    \item  Let us first show that the reduced instance of \sigmaDdecprob is a yes-instance if the \textsc{3-Partition} instance is a yes-instance. 
    
    We assume that the \textsc{3-Partition} instance is a yes-instance: There exists a partition of the integers of \setInt into \nbTriplet triplets $\nameTriplet_1$ to $\nameTriplet_\nbTriplet$. 
We schedule the tasks corresponding to the integers in each triplet within the separator sets. Specifically, if $x_i$, $x_j$, and $x_k$ are elements of triplet $\nameTriplet_b$, then the corresponding tasks $\nameTask_i$, $\nameTask_j$, and $\nameTask_k$ are scheduled together as a group. For the first half of the schedule — that is, for $1 \leq b \leq \nbTriplet/2$ — we place the tasks corresponding to the integers of triplet $\nameTriplet_b$ between \setSeparator{b-1} and \setSeparator{b}, ordered by non-increasing processing times. For the second half of the schedule — that is, for $\nbTriplet/2 + 1 \leq b \leq \nbTriplet$ — we schedule the tasks corresponding to the integers of triplet $\nameTriplet_b$ between \setSeparator{b} and \setSeparator{b+1}, in non-decreasing order of processing time. In other words, the task with the smallest processing time is scheduled first, followed by the second smallest, and finally the largest.

Now observe that, since each triplet sums exactly to $\sumTriplet$, the total processing time of the tasks scheduled between any pair of separators is precisely $\sumTripletNew$. This implies that every task belonging to a separator set is completed on time, and thus its deviation is exactly zero.
Next, we provide an upper bound on the deviation of the tasks in $\setIntTask$. We focus on the first half of the schedule, as the same reasoning applies symmetrically to the second half. In the worst case, all tasks corresponding to integers in the triplets of the first half have their median completion time located in the second half of the central time interval. We now divide this distance into subparts:

\begin{itemize}
    \item Firstly, all these tasks must be moved past \setSeparator{\nbTriplet/2}, that is, by \sizeSeparator{} units of time.
    
    \item Secondly, they must also be moved beyond the first half of the central time interval, i.e., by $\nbTriplet\sumTripletNew/2$ additional units of time.
    
    \item In the worst case, the smaller the task, the later it is scheduled by the third voter within the central time interval. Consequently, the smallest task must cross the entire second half of the interval, i.e., $\nbTriplet\sumTripletNew/2$ units of time; the second smallest task must cross the same distance minus the processing time of the smallest one, that is, at most $\nbTriplet\sumTripletNew/2 - \sumTripletNew/4$, and so on. In total, this contributes
    \[
        \sum_{i=0}^{3\nbTriplet/2 - 1} \left(\nbTriplet\frac{\sumTripletNew}{2} - i\frac{\sumTripletNew}{4}\right).
    \]
    
    \item Then, tasks are moved in groups of three until they reach the time interval corresponding to their respective triplets. For instance, one triplet has already reached its interval (between \setSeparator{\nbTriplet/2 - 1} and \setSeparator{\nbTriplet/2}), the next triplet must move by $\sumTripletNew + \sizeSeparator$, the following one by twice this amount, and so on. In total, this adds
    \[
        3\sum_{i=0}^{\nbTriplet/2 - 1} i(\sumTripletNew + \sizeSeparator).
    \]
    
    \item Finally, once each triplet has reached its corresponding interval, the longest task must start first within the triplet. Hence, it completes at least $\sumTripletNew/3$ units of time after the beginning of the interval (if it were shorter, the triplet could not sum to $\sumTripletNew$). This corresponds to an additional distance of $2\sumTripletNew/3$ from the median completion time.
    
    The second task in the triplet completes at the earliest $2\sumTripletNew/3$ units after the start of the interval (otherwise, the last task would exceed $\sumTripletNew/3$ in length, violating the non-increasing order of processing times). This adds another $\sumTripletNew/3$ to the distance, while the final task completes exactly at the end of the interval.
    
    Therefore, each triplet contributes an additional total distance of $\sumTripletNew$.
\end{itemize}



This provides an upper bound on the cost of an existing solution when a \textsc{3-Partition} solution exists, which we denote by $D_{3\text{-}P,UB}$:

\[
D_{3\text{-}P,UB} = D_{NF} + 2\Bigg(
    \frac{3\nbTriplet}{2} \big(\sizeSeparator + \frac{\nbTriplet \sumTripletNew}{2}\big)
    + \frac{\nbTriplet}{2} \sumTripletNew
    + \sum_{i=0}^{3\nbTriplet/2 - 1} \Big(\frac{\nbTriplet \sumTripletNew}{2} - i\frac{\sumTripletNew}{4}\Big)
    + 3 \sum_{i=0}^{\nbTriplet/2 - 1} i (\sumTripletNew + \sizeSeparator)
\Bigg)
\]

We can simplify this expression as follows:

\[
D_{3\text{-}P,UB} = D_{NF} + \frac{51 \nbTriplet^2 \sumTripletNew}{16} - \frac{\nbTriplet \sumTripletNew}{8} + \frac{3 \nbTriplet^2 \sizeSeparator}{4} + \frac{3 \nbTriplet \sizeSeparator}{2}.
\]




Since this is an upper bound and $\targetSum = D_{3\text{-}P,UB}$, it follows that if the instance of \textsc{3-Partition} is a \textit{yes}-instance, then the corresponding reduced instance of \sigmaDdecprob is also a \textit{yes}-instance.

\item  Let us now show that the reduced instance of \sigmaDdecprob is a yes-instance only if the \textsc{3-Partition} instance is a yes-instance.


 We note here that 
$D_{NF} + \frac{51 \nbTriplet^2 \sumTripletNew}{16} - \frac{\nbTriplet \sumTripletNew}{8} + \frac{3 \nbTriplet^2 \sizeSeparator}{4} + \frac{3 \nbTriplet \sizeSeparator}{2}$
is strictly smaller than $D_{NF} + 2 \factorK \sizeSeparator$. Indeed, since $\factorK = 4 \Big\lceil \frac{3}{4} \nbTriplet^2 + \frac{3}{2} \nbTriplet \Big\rceil$, we have $\factorK \sizeSeparator \geq \frac{3 \nbTriplet^2 \sizeSeparator}{4} + \frac{3 \nbTriplet \sizeSeparator}{2}$. Similarly, since $\sizeSeparator = 2 \Big\lceil \frac{51}{16} \nbTriplet^2 \sumTripletNew + \frac{\nbTriplet \sumTripletNew}{8} \Big\rceil$, it follows that $\factorK \sizeSeparator > \frac{51 \nbTriplet^2 \sumTripletNew}{16} - \frac{\nbTriplet \sumTripletNew}{8}$. Hence, we conclude that $D_{NF} + 2 \factorK \sizeSeparator > \targetSum$.


 Let us now assume that the reduced instance of \sigmaDdecprob\ is a \textit{yes}-instance, i.e., there exists an \textbf{optimal} schedule $\sched^*$ with deviation at most $\targetSum$. We now show that the tasks scheduled between the separator sets are necessarily grouped into triplets whose total processing time is exactly $\sumTripletNew$, except for the interval between \setSeparator{\nbTriplet/2} and \setSeparator{\nbTriplet/2+1}, which contains only tasks from \setCentral.

Let us assume, for the sake of contradiction, that there exists a time interval between two separator sets (except between \setSeparator{\nbTriplet/2} and \setSeparator{\nbTriplet/2+1}) that does not contain a triplet of tasks whose processing times sum exactly to $\sumTripletNew$.

We consider three subcases: (1) there exists an interval between two separator sets occupied only by a triplet of tasks from \setIntTask, but their processing times do not sum to $\sumTripletNew$; (2) no interval satisfies the condition of case (1), and there exists an interval containing a triplet of tasks from \setIntTask together with tasks from \setCentral such that their total processing time is strictly larger than $\sumTripletNew$; (3) all intervals between separator sets are occupied either by a triplet of tasks from \setIntTask with total processing time $\sumTripletNew$, or by a set of tasks from \setIntTask and \setCentral with total processing time at most $\sumTripletNew$. In all cases, we either find a lower bound $D_{N-3-P,LB}$ on the total deviation of $\sched^*$ that is incompatible with the deviation being smaller than $\targetSum$, or show that $\sched^*$ is not optimal. For the rest of the proof, we assume $i < \nbTriplet/2$, but the argument is symmetric if $i > \nbTriplet/2$. By symmetry, ``the first'' interval, i.e., the first from \setSeparator{0}, becomes the last, i.e., the first starting from \setSeparator{\nbTriplet+2}; separator sets \setSeparator{i} and \setSeparator{i+1} become \setSeparator{i} and \setSeparator{i-1}; and tasks from \setSeparator{i+1} being early (resp. late) correspond to tasks from \setSeparator{i-1} being late (resp. early).

\begin{itemize}

    \item (1) There exist two separator sets such that $\sched^*$ schedules three tasks from \setIntTask whose total processing time differs from $\sumTripletNew$. Consider the first pair of separators, \setSeparator{i} and \setSeparator{i+1}, satisfying this condition. Since these tasks do not form a triplet summing to $\sumTripletNew$ (as tasks in \setIntTask have processing times between $\sumTripletNew/4$ and $\sumTripletNew/2$, any set with cardinality different from 3 cannot sum to $\sumTripletNew$), their total processing time differs from $\sumTripletNew$.  

All processing times in \setIntTask are obtained by multiplying an integer $x_j$ by $\factorK$, so the sum can be written as $w \factorK$ for some integer $w \neq \sumTriplet$. Therefore, the total processing time differs by at least $\factorK$ units, implying that the $\sizeSeparator$ tasks of \setSeparator{i+1} are at least $\factorK$ units away from their preferred completion times. Hence, we can lower bound the total deviation of $\sched^*$ as
$$
D_{N-3-P,LB} = D_{NF} + 3 \factorK \sizeSeparator.
$$
Since $\factorK, \sizeSeparator > 1$ and $D_{NF} + 2 \factorK \sizeSeparator > \targetSum$, this leads to a contradiction.


\item (2) No interval satisfies case (1), but there exists an interval between two separator sets containing a triplet of tasks from \setIntTask together with tasks from \setCentral, such that their total processing time exceeds $\sumTripletNew$. Consider the first pair of separators, \setSeparator{i} and \setSeparator{i+1}, satisfying this condition. The tasks from \setIntTask have total processing time at least $\sumTripletNew$. In this case, the tasks of \setSeparator{i+1} are necessarily late, since the tasks from \setSeparator{i} were completed on time and the total processing time in between exceeds $\sumTripletNew$.  

By swapping one of the \setCentral tasks to after \setSeparator{i+1}, the deviation of the \setSeparator{i+1} tasks is reduced by $3 \sizeSeparator$, and the swapped \setCentral task also gets closer to its median completion time. This lowers the total deviation, contradicting the optimality of $\sched^*$.


\item (3) All time intervals between separator sets are occupied either by a triplet of tasks from \setIntTask summing to $\sumTripletNew$, or by a set of tasks from \setIntTask and \setCentral with total processing time at most $\sumTripletNew$. We distinguish two subcases:  

\begin{itemize}
    \item (3.1) There are $\nbTriplet-1$ intervals occupied by triplets of tasks from \setIntTask summing exactly to $\sumTripletNew$.
    \item (3.2) There are at least two intervals containing a mixture of tasks from \setIntTask and \setCentral.
\end{itemize}

Let us now consider each of these two cases.
    
    \begin{itemize}
        %

\item (3.1) If there are $\nbTriplet-1$ intervals containing triplets of tasks from \setIntTask, each with total processing time $\sumTripletNew$, then the last three tasks of \setIntTask must also sum exactly to $\sumTripletNew$, since the total processing time of all tasks from \setIntTask is precisely $\nbTriplet \sumTripletNew$. By assumption, the last interval cannot contain a full triplet of total processing time $\sumTripletNew$, so it contains at most two tasks from~\setIntTask.  

The tasks from \setSeparator{i+1} are early by the difference between $\sumTripletNew$ and the processing time of the two tasks from \setIntTask, which is at least $\sumTripletNew/4$, minus the number of tasks from \setCentral in the interval. These tasks therefore deviate by at least $\sizeSeparator$ units from their median completion time. Let $y$ denote the number of tasks from \setCentral in the interval. Then we can lower bound the deviation by
$$
D_{N-3-P,LB} \geq D_{NF} + 3 (\sumTripletNew/4 - y) \sizeSeparator + y \sizeSeparator.
$$
This expression is minimized when $y = \sumTripletNew/4$, giving
$$
D_{N-3-P,LB} \geq D_{NF} + \frac{\sumTripletNew \sizeSeparator}{4} \geq D_{NF} + 2 \factorK \sizeSeparator,
$$
which is a contradiction.




        \item (3.2) There are at least two time intervals containing both tasks from \setCentral and \setIntTask. Using the same argument than in case (3.1), the set of task of \setIntTask has a total processing time of at most $\factorK(\sumTriplet-1)$. The tasks from set \setSeparator{i+1} are early or on time, since we are not in case (2). Now, if there are $y$ tasks from \setCentral, it means that tasks from the second separator sets are early by at least $\factorK-y$ units of time. We can then lower bound the deviation caused by each of these interval by $3(\factorK-y)\sizeSeparator+y(\sizeSeparator)$. Clearly this is minimized by $y=\factorK$, meaning that for each of these intervals, the deviation is increased by at least $\factorK\sizeSeparator$, as there are at least two of these intervals, the total deviation is lower bounded by 
        $D_{N-3-P,LB}\geq D_{NF}+2\factorK\sizeSeparator$,  
        which is strictly larger than \targetSum, a contradiction.
        Note that this also holds if one of the two intervals is in the first half and the other in the second half of the schedule.
    \end{itemize}
\end{itemize}
\end{itemize}
This completes the proof.
\end{proof}

While the reduction with 4 voters can be easily adapted for the Spearman footrule extension of~\citeauthor{kumar2010generalized}~\cite{kumar2010generalized}, this is not the case for the reduction with 3 voters. This is because the second reduction uses the fact that it is more expensive, it terms of deviation, to move $p$ tasks of length 1 than to move 1 task of length $p$ from its median completion time. This does not hold when these costs are weighted by the length of the tasks. If one allows 5 voters (or more), it should be possible to modify this reduction by adding a copy of the first voter and a copy of the second. In such a situation moving a task from $\setCentral$ from its median completion time would cost $3$ times more. Indeed, the task would now move away for 4 voters and closer for 1, instead of away for 2 and closer for 1. On the other hand, moving a task from \setTask away from its median completion time would have the same cost, since it would go from away for 2, closer for 1, to away for 3, closer for 2. With these costs it should again be possible to find a value for \targetSum such that the only schedules meeting the \targetSum requirement are precisely the schedules corresponding to solutions for \textsc{3-Partition}. This leaves open the question of the complexity with exactly three voters for the definition of \citeauthor{kumar2010generalized}\cite{kumar2010generalized}.

\section{Conclusion}
\label{sec:conclu}

We investigated the complexity of an extension of the Spearman footrule with a small fixed number of voters, showing that it is NP-hard for 3 voters and 4 voters, our reduction with 4 voters can also be applied to another existing extension of Spearman footrule~\citep{kumar2010generalized}. 
This strong complexity guarantee could be balanced with positive results using parameterized complexity theory. While our results rule out the existence of XP algorithms parameterized only by the number of voters for the corresponding problems, one can see that the problem is FPT with respect to the number of tasks, as we only look for a permutation of the tasks. The hardness results in this paper rely on the length of the tasks, one promising parameter could be the maximum length of a task. It is also commonly used for parameterized algorithm in scheduling~\cite{mnich2018parameterized}. Interestingly, our reductions also rely on the fact that some tasks are scheduled  at the beginning of the schedule for one voter and at the very end by some other voter. Thus, it could be interesting to see whether the problem becomes easier when parameterized by the maximum difference between the maximum and minimum completion times of a task in the profile, a similar parameter, sometimes denoted as range of a candidate, exists in computational social choice~\cite{betzler2009fixed}.

\bibliographystyle{plainnat}
\bibliography{bibliography}

@book{brucker1999scheduling,
author = {Brucker, Peter},
title = {Scheduling Algorithms},
year = {2010},
isbn = {3642089070},
publisher = {Springer},
edition = {5th}
}

@book{brandt2016handbook,
  title={Handbook of computational social choice},
  author={Brandt, Felix and Conitzer, Vincent and Endriss, Ulle and Lang, J{\'e}r{\^o}me and Procaccia, Ariel D},
  year={2016},
  publisher={Cambridge University Press}
}

@inproceedings{pascual2018collective,
  TITLE = {{Collective Schedules: Scheduling Meets Computational Social Choice}},
  AUTHOR = {Pascual, Fanny and Rzadca, Krzysztof and Skowron, Piotr},

  BOOKTITLE = {{Seventeenth International Conference on Autonomous Agents and Multiagent Systems }},
  YEAR = {2018},
  MONTH = Jul,
  HAL_ID = {hal-01744728},
  HAL_VERSION = {v1},
}

@article{wan2013single,
  title={Single-machine scheduling to minimize the total earliness and tardiness is strongly NP-hard},
  author={Wan, Long and Yuan, Jinjiang},
  journal={Operations Research Letters},
  volume={41},
  number={4},
  pages={363--365},
  year={2013},
  publisher={Elsevier}
}

@article{kemeny1959mathematics,
  title={Mathematics without numbers},
  author={Kemeny, John G},
  journal={Daedalus},
  volume={88},
  number={4},
  pages={577--591},
  year={1959},
  publisher={JSTOR}
}

@article{diaconis1977spearman,
  title={Spearman's footrule as a measure of disarray},
  author={Diaconis, Persi and Graham, Ronald L},
  journal={Journal of the Royal Statistical Society: Series B (Methodological)},
  volume={39},
  number={2},
  pages={262--268},
  year={1977},
  publisher={Wiley Online Library}
}

@article{bartholdi1989voting,
  title={Voting schemes for which it can be difficult to tell who won the election},
  author={Bartholdi, John and Tovey, Craig A and Trick, Michael A},
  journal={Social Choice and welfare},
  volume={6},
  number={2},
  pages={157--165},
  year={1989},
  publisher={Springer}
}

@inproceedings{kumar2010generalized,
  title={Generalized distances between rankings},
  author={Kumar, Ravi and Vassilvitskii, Sergei},
  booktitle={Proceedings of the 19th international conference on World wide web},
  pages={571--580},
  year={2010}
}

@article{EdmondsK72,
  author       = {Jack Edmonds and
                  Richard M. Karp},
  title        = {Theoretical Improvements in Algorithmic Efficiency for Network Flow
                  Problems},
  journal      = {J. {ACM}},
  volume       = {19},
  number       = {2},
  pages        = {248--264},
  year         = {1972},
  bibsource    = {dblp computer science bibliography, https://dblp.org}
}

@article{Tomizawa71,
  author       = {N. Tomizawa},
  title        = {On some techniques useful for solution of transportation network problems},
  journal      = {Networks},
  volume       = {1},
  number       = {2},
  pages        = {173--194},
  year         = {1971},
  bibsource    = {dblp computer science bibliography, https://dblp.org}
}

@book{garey2002computers,
  title={Computers and intractability},
  author={Garey, Michael R and Johnson, David S},
  volume={29},
  year={2002},
  publisher={wh freeman New York}
}

@InProceedings{durand2022collective,
author="Durand, Martin
and Pascual, Fanny",
editor="Kanellopoulos, Panagiotis
and Kyropoulou, Maria
and Voudouris, Alexandros",
title="Collective Schedules: Axioms and Algorithms",
booktitle="Algorithmic Game Theory",
year="2022",
publisher="Springer International Publishing",
address="Cham",
pages="454--471",
isbn="978-3-031-15714-1"
}

@article{mnich2018parameterized,
  title={Parameterized complexity of machine scheduling: 15 open problems},
  author={Mnich, Matthias and Van Bevern, Ren{\'e}},
  journal={Computers \& Operations Research},
  volume={100},
  pages={254--261},
  year={2018},
  publisher={Elsevier}
}

@article{betzler2009fixed,
  title={Fixed-parameter algorithms for Kemeny rankings},
  author={Betzler, Nadja and Fellows, Michael R and Guo, Jiong and Niedermeier, Rolf and Rosamond, Frances A},
  journal={Theoretical Computer Science},
  volume={410},
  number={45},
  pages={4554--4570},
  year={2009},
  publisher={Elsevier}
}

\end{document}